\documentclass[12pt,a4paper]{article}
\usepackage{lineno}
\usepackage[utf8]{inputenc}
\usepackage[english]{babel}
\usepackage{amsmath}
\usepackage{amsfonts}
\usepackage{amssymb}
\usepackage{graphicx}
\usepackage{amsthm}
\usepackage{slashbox}
\usepackage{float}
\usepackage{subfig}
\usepackage[font=footnotesize]{caption}
\usepackage[left=2.5cm,right=2.5cm,top=2.5cm,bottom=2.5cm]{geometry}
\usepackage{setspace}
\author{Pier Luigi Conti \\ Alberto Di Iorio}
\date{}
\newcommand{\U}{\mathcal{U}_N}
\newcommand{\Y}{\mathcal{Y}_N}
\newcommand{\T}{\mathcal{T}_N}
\newcommand{\E}{\mathbb{E}}
\newcommand{\MP}{\mathbb{P}}
\newcommand{\D}{\mathbf{D}_N}
\usepackage{amsthm}

\newtheorem{Lemma}{Lemma}
\newtheorem{prop}{Proposition}
\newtheorem{cor}{Corollary}
\newtheorem{rem}{Remark}

\title{Analytic inference in finite population framework via resampling}
\begin{document}
\begin{titlepage}
\clearpage\maketitle
\thispagestyle{empty}

\begin{abstract}
The aim of this paper is to provide a resampling technique that allows us to make inference on superpopulation parameters
in finite population setting.
Under complex sampling designs, it is often difficult to obtain explicit results about superpopulation parameters of interest,
especially in terms of confidence intervals and test-statistics.
Computer intensive procedures, such as resampling, allow us to avoid this problem.
To reach the above goal, asymptotic results about empirical processes in finite population framework are first obtained.
Then, a resampling procedure is proposed, and justified {\em via} asymptotic considerations.
Finally, the results obtained are applied to different inferential problems and a simulation study is performed to test the goodness of our proposal.


\end{abstract}
\end{titlepage}

\onehalfspacing

\section{Introduction}
\label{s1}
The use of superpopulation models in survey sampling has a long history, going back (at least) to
\cite{r1}, where the limits of assuming the population characteristics as {\em fixed}, especially in economic and social studies,
are stressed.
As clearly appears, for instance, from \cite{r2} and \cite{pfeffer93}, there are basically two types of inference in the finite populations setting.
The first one is \textit{descriptive} or \textit{enumerative} inference, namely inference about finite population parameters.
This kind of inference is a static ``picture'' on the current state of a population, and does not take into account the mechanism generating the
characters of interest of the population itself.
The second one is \textit{analytic} inference, and consists in inference on superpopulation parameters.
This kind of inference is about the process that generates the finite population. In contrast  with \textit{enumerative} inference results, \textit{analytic} ones are more general, and
still valid for every finite population generated by the same superpopulation model.

The present paper essentially focuses on analytic inference for nonparametric superpopulation models. In classical (nonparametric) statistics,
under the Fisherian inferential  framework, a popular approach consists in approximating the distribution of estimators and test-statistics {\em via} bootstrap (cfr. \cite{r3}, \cite{romano88}, \cite{romano89} and
references therein).
Efron's bootstrap procedure (\cite{r3}) is based on a crucial assumption: data are independent and identically distributed ($i.i.d$).
Unfortunately, this is not the case of finite population framework, where the presence of a complex sampling design induces dependences in the data.
For this reason, several different resampling techniques in finite populations setting have been proposed in the literature.

A large portion of such techniques essentially refers to descriptive inference, and
rests on the idea of mimicking the moments of the sampling distributions of statistic of interest.
In particular, in case of Horvitz-Thompson estimator of the population mean,
this idea reduces to require that the variance of the resampled statistic should be equal (or at least very close)
to the variance estimate of the original statistic.
This is usual attained by resampling units according to some special sampling design that takes into account the dependence between units: cfr. \cite{r10} and references therein.

The arguments above are considerably different from those commonly used to justify the classical bootstrap, that are based on asymptotic considerations involving the whole sampling distribution of a statistic, not only the first two moments.
In particular, in \cite{r4}, usual Efron's bootstrap is justified by proving that the asymptotic distribution of a bootstrapped statistic
coincides with that of the original statistic.
In our knowledge, the only papers that develop resampling methods for finite populations justified {\em via} asymptotic arguments are \cite{chatter11}, \cite{r5}, \cite{contmarmec}.
All the above mentioned papers are based on the fixed population approach, {\em i.e.} refer to the estimation of finite population parameters (descriptive inference).
Furthermore, \cite{chatter11} is confined to quantile stimation under simple random sampling. The results are then extended to general $\pi$ps designs in \cite{r5}.

In \cite{contmarmec} a class of resampling procedures based on a predictive approach is defined, and their asymptotic distribution is studied.
Such procedures are essentially taylored for the estimation of finite population parameters, in a descriptive inference perspective.
In the present paper, we will generalize the results in  \cite{contmarmec} to analytic inference. As it will be seen in the sequel, the analytic-inference perspective
dramatically changes the asymptotic distributions to be considered. As a consequence, the resampling procedures defined in  \cite{contmarmec} do not work when
superpopulation parameters are involved; the only exception is the so-called ``multinomial'' approach, defined first in \cite{pfefsver04}.

In a recent paper by \cite{wang}, and more rigorously in \cite{r24}, the authors obtain a result substantially equivalent to Proposition $\ref{main}$. However, they have the only purpose of establishing a functional central limit theorem, without proposing a resampling scheme that allows to recover the large sample distribution of  statistics of interest.
On one hand, the regularity assumptions in \cite{r24} are slightly weaker than ours; on the other hand, they assume the asymptotic normality of the distribution function estimator for the considered sampling design, while in the present paper this assumption is avoided, and replaced by the {\em high entropy} condition for the considered sampling designs.
Such a condition, although slightly more restrictive than those in \cite{r24}, allows us to explicitly write down the covariance kernel function of the asymptotic law of the considered functionals, without resorting to the computation of second order inclusion probabilities, usually a numerically complicate task for almost all $\pi$ps sampling designs.
Moreover,
approximations of second order inclusion probabilities, such as H\'{a}jek approximation (cfr. \cite{r16}) essentially work only for high entropy sampling designs.

The paper is organized as follows. In Section $\ref{s2}$ the assumptions on which the paper rests are stated. Sections $\ref{s3}$, $\ref{s4}$ are devoted to
 establish the main asymptotic results for a large class of estimators. In Sections $\ref{s5}$, $\ref{s5b}$, the proposed resampling procedure and its asymptotic justification
 are studied.
 Finally, in Section $\ref{s6}$, some applications ale illustrated, and studied {\em via} simulation.

\section{Assumptions and basics}\label{s2}
Let $\U$ be a finite population of size $N$ and $s \subset \U$ a sample of size $n_s$. For each unit in the population, denote by
$$D_i=\begin{cases}1 &\mbox{if unit }i \in s\\
 0 &\mbox{otherwise}
 \end{cases}$$
 the sample inclusion (Bernoulli) random variable (r.v.), and let  $\mathbf{D}_N$ be the vector composed by the $N$ random variables $D_1,\ldots,D_N$.
The probability distribution $P$ of the r.v. $\D$  is the sampling design.\\
 For each $i,j \in \U$ the moments $\pi_i=\E_P[D_i]$ and $\pi_{ij} = \E_P[D_iD_j]$ are the first and second order inclusion probabilities. The sum  $n_s=D_1+D_2+\ldots +D_N$ is the effective sample size; in the sequel we will focus on fixed size sampling designs: $n_s\equiv n$.\\
A Poisson design (denoted by $Po$) with parameters $p_1,\ldots,p_N$ has mass function equal to:
$$
 {Po}(\mathbf{D}_N)=\prod_{i=1}^N{p_i}^{D_i}(1-p_i)^{1-D_i}.
$$
The next basic sampling design we consider is the rejective sampling,
denoted by the symbol $R$. Rejective sampling is essentially a Poisson sampling conditioned on a fixed sample size (for more see \cite{r15}).\\
A measure of the randomness of a sampling design $P$ is its entropy:
$$H(P)=\E_P[\log P(\mathbf{D}_N)]=-\sum_{\mathbf{D}_N} P(\mathbf{D}_N)\log(P(\mathbf{D}_N)).$$
Is well known that the Poisson sampling possesses Maximum Entropy among sampling designs with fixed first order inclusion probabilities. The rejective sampling, being strongly related to the Poisson sampling, inherits this property, and it is possible to show (cfr. \cite{r16}) that it is the maximum entropy design among sampling designs of fixed size and fixed first order inclusion probabilities.\\
To quantify the similarity between a generic sampling designs $P$ and the rejective design $R$ we use the Hellinger distance, defined as
\begin{equation}
d_H(P,R)=\sum_{\D}\left(\sqrt{P(\D)}- \sqrt{R(\D)} \right)^2.
\end{equation}
The basic assumptions on which all subsequent results  rest are listed below.
\begin{itemize}
\item[H1.]$(\U;\ N\geq 1)$ is a sequence of finite populations of increasing size $N$.
\item[H2.]Let $Y$ be the character of interest, and let $T_1,T_2,\ldots,T_L$ be the design variables. Denote further by $\mathbb{P}$ the superpopulation proability distribution of the r.v.s
$(Y_i,T_{i1},$ $\ldots,$ $T_{iL}).$
For each size $N$, $(y_i,t_{i1},\ldots,t_{iL}),$ $i=1,2,\ldots,N$ are realizations of a superpopulation $\{(Y_i,T_{i1},$ $\ldots,$
$T_{iL}), \ i=1,\ldots,N\}$ composed by $i.i.d$ $(L+1)$-dimensional random vectors.
The symbols $\Y$, $\T$ are used to denote the vector of $N$ population $y_i$s values and the $N \times L$ matrix of
population $t_{ij}$s values ($j=1, \, \dots , \, L$), respectively.
\item[H3.]For each population $\U$, sample units are selected according to a fixed size sample design with positive first order inclusion probabilities $\pi_1,\ldots,\pi_N$ and sample size $n=\pi_1+\ldots+\pi_N$. The first order inclusion probabilities are taken proportional to a variable $x_i=g(
t_{i1},\ldots,t_{iL})$, $i=1,\ldots,N$, where $g(\cdot)$ is an arbitrary positive function. For the sake of simplicity, we will assume that, for each $i,$ $\pi_i=nx_i/\textstyle \sum_j x_j$.
Clearly, the quantities $n,\pi_i,D_i$ depend on $N$. To avoid complications in the notation we will use the symbols $n,\pi_i,D_i$, omitting the explicit dependence on $N$.
Furthermore  is assumed that
\begin{equation}
\E_{\MP}[\pi_i(1-\pi_i)]=d \label{L1}
\end{equation}
with $0<d<\infty$.
\item[H4.]The sampling fraction tends to a finite, non-zero limit:
$$\lim_{N\to\infty}\frac{n}{N}=f, \ 0<f<1.$$
\item[H5.]The actual sampling design $P$, with inclusion probabilities $\pi_1,\ldots,\pi_N$ satisfies the relationship
$$d_H(P,R)\to0, \mbox{ as } N\to\infty,$$
where $R$ is the rejective sampling with the same inclusion probabilities as $P$.
\item[H6.] $\E_\MP[X_1^2]<\infty$.
\end{itemize}
Hypothesis H2, H3 allow us to consider a possible dependence between the interest variable and the design variables. This is the usual situation when we deal with  $\pi ps$ sampling designs, where such a dependence is used to improve the efficiency of total and mean estimators. On the other hand, the specific form of dependence is totally general.\\
Assumption H5 essentially requires that the considered sampling design has to be an asymptotically high entropy sampling designs. The properties of high entropy sampling designs are widely discussed in literature; see, for instance, \cite{r15}, \cite{r17}, \cite{r18}. One of these properties is that high entropy sampling designs with the same inclusion probabilities, have the same asymptotic behaviour and it depends only on first order inclusion probabilities.\\
From now on, we will denote by $F(y)$ the superpopulation distribution function of the variable of interest $Y$, by $G(x)$ the distribution function of the design variable $X
= g( T_1 , \, \dots , \, T_L )$, and by $H(x,y)$ the joint distribution function of the r.v. $(X,Y)$.\\
The finite population distribution function is defined as:
\begin{equation}
F_N(y)=\dfrac{1}{N}\sum_{i=1}^NI_{(y_i\leq y)}
\end{equation}
where $I_{(y_i \leq y)}$ is the usual indicator taking value 1 if $y_i$ lies in $(-\infty,y]$, and $ 0$ otherwise.\\
A superpopulation parameter (hyperparameter, for short) is a functional of $F$, namely:
\begin{equation}
\label{parameter}
\theta=\theta(F)
\end{equation}
One of the most used and intuitive approaches to estimate a hyperparameter of the form (\ref{parameter}) consists in replacing $F$ in (\ref{parameter}) by an appropriate estimate.\\
As an estimator of $F$ we consider here the H\'{a}jek ratio estimator:
\begin{equation}\label{Hajek}
\hat{F}_H(y)=\dfrac{\displaystyle\sum_{i=1}^N \dfrac{1}{\pi_i}D_iI_{(y_i\leq y)}}{\displaystyle\sum_{i=1}^N\dfrac{1}{\pi_i}D_i}.
\end{equation}
Before ending the present section, we point out that all results
of the subsequent sections
could be obtained, with minor variations,  by using the Horvitz-Thompson estimator of $F$
\begin{equation}\label{HTF}
\hat{F}_{HT}(y)= \frac{1}{N} \sum_{i=1}^N \dfrac{1}{\pi_i} D_i I_{(y_i\leq y)} .
\end{equation}
However, we prefer the H\'{a}jek  estimator since $(\ref{Hajek})$ it is a proper estimator of the distribution function $F$.
The same is not generally true for the Horvitz-Thompson estimator $( \ref{HTF} )$.

\section{Empirical process in finite population sampling: asymptotic results}\label{s3}
The aim of this section is to study the limiting distribution of that H\'{a}jek  estimator (\ref{Hajek}) under
both the source of  randomness due to the sample selection and the source of randomness due to the population generation.
To this purpose, we have to study the stochastic process $W_H(\cdot)=(W_H(y),\ y \in \mathbb{R})$, defined as
\begin{equation}\label{Proc}
W^H(y)=\sqrt{n}(\hat{F}_H(y)-F(y)), \ y\in \mathbb{R}
\end{equation}
The process (\ref{Proc}) can be partitioned into the sum of two stochastic processes
\begin{equation}\label{Proc_dec}
\underbrace{\sqrt{n}(\widehat{F}_H-F)}_{\substack{\text{Total} \\\text{Randomness }}}=
\underbrace{\sqrt{n}(\widehat{F}_H-F_N)}_{\substack{\text{Sampling}\\ \text{Randomness}}}+\sqrt{\frac{n}{N}}\underbrace{\sqrt{N}(F_N-F)}_{
\substack{\text{Superpopulation}\\ \text{Randomness}}}=W_n^H+\sqrt{\frac{n}{N}}W_N,
\end{equation}
where
\begin{eqnarray}
W_n^H(y) = \sqrt{n}(\widehat{F}_H (y) -F_N (y) ) , \ y\in \mathbb{R}
\label{eq:sampling_rand}
\end{eqnarray}
depends on the sampling design (the sample selection randomness), and
\begin{eqnarray}
W_N(y) \sqrt{N} (F_N (y) - F(y)) , \ y\in \mathbb{R}
\label{eq:superpop_rand}
\end{eqnarray}
is a classical empirical process, and depends on the data generating process (superpopulation randomness).
In the sequel, we will refer to the process $(\ref{Proc})$ as an \textit{empirical process in finite population sampling}.\\
In the present section we will establish the asymptotic law of the empirical process $(\ref{Proc})$.
As it will be seen in Proposition $\ref{main}$, the limiting law of $(\ref{Proc})$ is
different from the asymptotic law of the usual empirical process for $i.i.d.$ data.
In our case the sample data are neither independent nor identically distributed
due to the effect of the sampling design, and this affects the asymptotic behaviour of  $(\ref{Proc})$. In addition,
the limiting law of $(\ref{Proc})$ heavily depends on the possible dependence between the character of interest and the design variables.

The limiting law of the process $(\ref{eq:sampling_rand})$, conditionally on $y_i$s and $t_{ij}$s, $(j=1, \, \dots , \, L)$, is studied in \cite{contmarmec}, where
Claim $1$ is proved.
Denote, as usual, by $D[-\infty,\infty]$
the space of \textit{càdlàg} (continue à droite, limite à gauche) functions defined on the (extended) real line, endowed with the Skorokhod topology. The compact sentence
``for almost all $y_i$s, $t_{ij}$s'' means ``for a set of sequences of $y_i$ and $t_{ij}$ values that are generated by the superpopulation model with
$\mathbb{P}$-probability 1''.

\begin{prop}\label{main}
Assume the sampling design $P$ satisfies conditions $H1-H6$. Then, the following three claims hold.

\noindent {\rm{Claim}} $1$ $({\mathrm{Conditional \; convergence}})$
Conditionally of $y_i$s, $t_{ij}$s ($j=1, \, \dots , \, L$), and for almost all
$y_i$s, $t_{ij}$s,
the sequence of random functions $(W_n^H(\cdot),\ N\geq 1)$ converges weakly in $D[ -\infty,\infty ]$
equipped with the Skorokhod topology, to a Gaussian process $\widetilde{W}_1(\cdot)=(\widetilde{W}_1(y),\ y\in \mathbb{R})$ with zero mean function and covariance kernel
\begin{align} \label{C1}
\nonumber C_1(y,t)&=f\left\lbrace \frac{\mathbb{E}_{\MP}[X_1]}{f}K_{-1}(y\wedge t) -1\right\rbrace F(y\wedge t)-\frac{f^3}{d}
\left(1-\frac{K_{+1}(y)}{\mathbb{E}_{\MP}[X_1]} \right)\left(1-\frac{K_{+1}(t)}{\mathbb{E}_{\MP}[X_1]} \right)F(y)F(t)\\ &-
f\left\lbrace \frac{\mathbb{E}_{\MP}[X_1]}{f}(K_{-1}(y)+K_{-1}(t)-\mathbb{E}_{\MP}[X_1^{-1}]-1) \right\rbrace F(y)F(t)
\end{align}
with $d$ given by $(\ref{L1})$, and
\begin{eqnarray}
K_{l} (y) = \mathbb{E}_{\MP} \left [ \left . X_i^{l} \, \right \vert Y_i \leq y \right ] , \;\;
y \in {\mathbb{R}}, \; l = 0, \, \pm 1, \, \pm 2 . \nonumber
\end{eqnarray}

\noindent {\rm{Claim}} $2$ $({\mathrm{Unconditional \; convergence}})$
The sequence of random functions $(W_n^H(\cdot),\ N\geq 1)$, converges weakly, in $D[-\infty,\infty]$
equipped with the Skorokhod topology, to a Gaussian process ${W}_1(\cdot)=({W}_1(y),\ y\in \mathbb{R})$ with zero mean function and covariance kernel
$( \ref{C1} )$.

\noindent {\rm{Claim}} $3$ $({\mathrm{Main \; result}})$
The two sequences $(W_n^H(y),\ y \in \mathbb{R})$ and $(W_N(y),\ y \in \mathbb{R})$ are asymptotically independent.
As a consequence, the whole process $(W^H(y),\ y \in \mathbb{R})$ converges weakly in $D[-\infty,\infty]$ endowed with the Skorokhod topology, to a Gaussian process $W$ with zero mean function and covariance kernel
\begin{equation}\label{Kernel}
C(y,t)=C_1(y,t)+fC_2(y,t)
\end{equation}
where $C_1(y,t)$ and $C_2(y,t)$ are given by $(\ref{C1})$ and
\begin{eqnarray}\label{C2}
C_2(y,t)=F(y\wedge t)-F(y)F(t) ,
\end{eqnarray}
respectively.
\end{prop}

We stress here that working conditionally on $y_i$s, $t_{ij}$s, is equivalent to consider the population is {\em fixed} (even if with increasing size), although generated by a superpopulation model. Hence, Claim 1 of Proposition $\ref{main}$ essentially refer to descriptive inference.
By the decomposition \eqref{Proc_dec}, it is clear that \rm{Claim} $1$ takes into consideration the contribution of the sampling design to the limit distribution of the process \eqref{Proc}, while \rm{Claim} $2$ takes into account the contribution of the superpopulation model to the limit distriution of the whole process \eqref{Proc} that is stated in \rm{Claim} 3.

A special case on which it is worth to focus is when the character of interest $Y$ and the design variable $T_j$s are independent,
that is essentially the case studied in  \cite{r6}.
In this case, the covariance kernel
(\ref{C1}) reduces to:
\begin{eqnarray}
C_1(y,t)=f(A-1)(F(y\wedge t)-F(y)F(t)) \nonumber
\end{eqnarray}
where
\begin{equation}\label{A}
A=\dfrac{\E_\MP[X_1]}{f}\E_\MP[X_1^{-1}]
\end{equation}
is, by the strong law of large numbers, the almost sure limit of
\begin{equation*}
\frac{1}{N}\sum_{i=1}^N \frac{1}{\pi_i}.
\end{equation*}
The following corollary sums up this result.
\begin{cor}\label{indep}
Under the hypothesis $H1-H6$, if $Y$ and $T_j$s are independent, the sequence $(W^H(y), \ y \in \mathbb{R})$ converges weakly, in $D[-\infty,\infty]$ equipped with the Skorokhod topology, to a Gaussian process with zero mean function and covariance kernel
\begin{eqnarray}
C(y,t)=fA(F(y\wedge t)-F(y)F(t)) , \nonumber
\end{eqnarray}
with $A$ given by $(\ref{A})$.
\end{cor}

The limiting process of Corollary $\ref{indep}$ is proportional to a Brownian bridge on the scale of $F$, which is
the usual limiting process of the empirical process in classic setting of $i.i.d$ data. The proportionality constant takes into account the finite population setting (the sampling fractions appears in the expression of the proportionality constant) and also  the dependence relationship between units due to the sampling design (the term $A$).

Another case of interest is when the sampling design is a simple random sampling (srs).
As shown in Corollary $\ref{Cor2}$, in this case the role of the sampling design is asymptotically negligible, and the unit in the sample can be seen as independently
selected by the superpopulation.
The following result formalizes this idea.
\begin{cor}\label{Cor2}
Under the hypothesis $H1-H6$, if the sampling design $P$ is a simple random sampling, the sequence $(W^H(y), \ y \in \mathbb{R})$ converges weakly, in $D[-\infty,\infty]$ equipped with the Skorokhod topology, to a Gaussian process with zero mean function and covariance kernel
\begin{eqnarray}
C(y,t)=(F(y\wedge t)-F(y)F(t)) . \nonumber
\end{eqnarray}
\end{cor}
It is easy to see that the H\'{a}jek  estimator (\ref{Hajek}) under a srs design coincides with the empirical distribution function of the sample.
 Hence, Corollary \ref{Cor2} states that the asymptotic law of the process $W^H$ under the srs is exactly a Brownian bridge as in the case of classical empirical processes under the $i.i.d$ data assumptions.

Next result, that will be used in Section $\ref{s5}$, is a Glivenko-Cantelli type result establishing the uniform convergence of $\widehat{F}_H$ to $F$.
\begin{prop}
\label{glivcant}
Under the hypotheses $H1-H6$, we have:
\begin{eqnarray}
\sup_{y} \left \vert \widehat{F}_H - F(y)  \right \vert \rightarrow 0 \;\; {\mathrm{as}} \; N \rightarrow \infty
\label{glivcantelli}
\end{eqnarray}
for a set of (sequences of) $Y_i$s, $T_{ij}$s having $\MP$-probability $1$, and for a set of $\D$s of $P$-probability tending to $1$ as $N$ increases.
\end{prop}

\begin{rem}
\label{rem_edf}
{\rm{Even if in the superpopulation model the r.v.s $Y_i$s are $i.i.d.$, the sampling design makes it inconsistent the common empirical distribution function (e.d.f.):
\begin{eqnarray}
\widehat{F}_n (y) = \frac{1}{n} \sum_{i=1}^{N} D_i I_{(y_i \leq y)} . \label{eq:edf_nonnorm}
\end{eqnarray}
In fact, it is not difficult to see that:
\begin{eqnarray}
\E_{\MP, P} \left [ \widehat{F}_n (y) \right ] & = & \frac{1}{n} \sum_{i=1}^{N} \E_{\MP} \left [
\pi_i I_{(y_i \leq y)}
\right ] \nonumber \\
\, & \rightarrow & \E_{\MP} \left [X I_{(Y \leq y)} \right ] /  \E_{\MP} \left [X \right ]  \neq F(y) \label{eq_edf_asint_bias}
\end{eqnarray}
\noindent as $N$ increases. Relationship $( \ref{eq_edf_asint_bias} )$ shows that the e.d.f. $( \ref{eq:edf_nonnorm} )$ is asymptotically
biased, and hence inconsistent.

The above result can be slightly refined. Using the same approach as in Lemma $\ref{lemma1}$, it is not difficult to show that, as $N$ increases,
\begin{eqnarray}
\widehat{F}_n (y) \rightarrow \E_{\MP} \left [X I_{(Y \leq y)} \right ] /  \E_{\MP} \left [X \right ]  \neq F(y) \nonumber
\end{eqnarray}
for a set of (sequences of) $y_i$s, $t_{ij}$s having $\MP$-probability $1$, and for a set of $\D$s of $P$-probability tending to $1$.
This makes it stronger the assertion about the inconsistency of $\widehat{F}_n (y) $, because it shows that such an inconsistency
is due to the sampling design.}}
\end{rem}

\section{Regularity assumptions to estimate hyperparameters}\label{s4}
As already said in Section \ref{s2}, we focus on hyperparameters ({\em i.e.} superpopulation parameters) that can be
expressed as functional of the d.f. $F$ of the character of interest $Y$.
The aim of this section is to introduce the proper regularity condition and to study the large sample distribution of estimators of superpopulation parameters.\\
The sought condition is the \textit{Hadamard}-differentiability, which is weaker than  \textit{Frechét} differentiability. In fact some well-known statistical functionals, like variance and quantiles (see \cite{r20}, p. 220, and \cite{r14}), do not satisfy the usual
\textit{Frechét} differentiability assumption.\\
Let $\theta(\cdot):l^\infty(-\infty,\infty)\to E$ be a map having as domain the Banach space (equipped with the sup-norm) of the bounded functions, and taking values on a normed space $E$ with norm $\Vert\cdot \Vert_E$. The map $\theta(\cdot)$ is \textit{Hadamard-differentiable} at $F$ if there exist a continuous linear functional $\theta'_F(\cdot):l^\infty(-\infty,\infty)\to E$ such that
\begin{equation}
\left\Vert \dfrac{\theta(F+th_t)-\theta(F)}{t}-\theta'_F(h)\right\Vert_E\rightarrow 0,\text{ as } \ t\downarrow0, \ \forall h_t\rightarrow h.
\end{equation}
The map $\theta'_F(\cdot)$ is the \textit{Hadamard derivative} of $\theta$ at $F$.

As a consequence of Theorem 20.8 (p. 297) in \cite{r9} and Proposition \ref{main} the following result holds true.
\begin{prop}\label{Had}
Suppose that $\theta(\cdot)$ is (continuously) Hadamard-differentiable at $F$, with Hadamard derivative $\theta'_F(\cdot)$. Assuming $H1-H6$, the sequence
$(\sqrt{n}(\theta(\hat{F}_H)-\theta(F)),\ y\in \mathbb{R})$ converges weakly to $\theta'_F(W)$, almost surely w.r.t. $\MP$, as $N$ increases.
\end{prop}
It is worth to analyse some consequences of Proposition \ref{Had}.
If $\theta(\cdot)$ takes value on the real line, the limiting random variable $\theta'_F(W)$ is Gaussian and centered; in fact the linearity of the Hadamard derivative preserve both normality and the zero mean function.
Thus, the variance of $\theta'_F(W)$ is equal to
\begin{equation}
\sigma^2_{\theta}=\E[\theta'_F(W)^2]
\end{equation}

\section{Resampling procedure: theoretical properties}\label{s5}
Computing the asymptotic distribution of the H\'{a}jek  estimator of
hyperparameters of interest requires the knowledge of the explicit form of the Hadamard-derivative of the functional. Sometimes this derivative is hard to compute, so the goal of this section is to provide a resampling procedure that allow us to recover the asymptotic distribution of the H\'{a}jek  estimator  avoiding the explicit computation of the Hadamard derivative of the functional under examination.\\
After defining the resampling procedure, we also provide a full asymptotic justification. The idea is similar to what proved for the classical bootstrap by
 \cite{r4}: the resampled process converges to the same limit of the original process.\\
A first attempt to define a resampling procedure justified by asymptotic considerations in finite populations framework is in \cite{r5}, and in \cite{contmarmec}.
In the present paper, there are several fundamental differences. First of all, both the above mentioned papers focus on descriptive inference, so that the involved asymptotic distributions are different. In the second place, in \cite{r5} there is asymptotically no relationship between the design variables and the variable of interest. The possible existence
of such a relationship is taken into account in \cite{contmarmec}, but, due to the descriptive inference framework, the class of resampling procedures defined in that paper do not work in
the present case, except the noticeable exception of the ``multinomial scheme'' described below.

The resampling procedure considered in the present paper in composed by two phases. In the first phase, a prediction of the population is generated on the basis of the sample. In the second phase, a new sample, of the same size of the original one, is selected according to a sampling design $P^*$ that fulfills the high entropy requirement. The inclusion probabilities are chosen proportional to the size variable $X$ of the predicted population constructed in Phase 1.
\begin{itemize}
\item[Phase 1.]
\medskip
\begin{itemize}
\item[1.] Sample $N$ units independently from the distribution $\hat{F}_H$, such that each unit $i \in s$ is selected with probability $\pi_i^{-1}/\sum_{j\in s}\pi_j^{-1}$ $=$ $\pi_i^{-1}/\sum_{j=1}^{N} D_j \pi_j^{-1}$
\item[2.] For $k=1,2,\ldots,N$, if the $k-th$ sampled unit is unit $i\in s$, take $y_k^*=y_i$ and $x_k^*=x_i$.
\item[3.] Define a predicted population of $N$ units $\U^*$, such that unit $k$ possesses $y$-value $y_k^*$ and $x$-value $x^*_k$, $k=1,2,\ldots,N$.
\end{itemize}
\item[Phase 2.]
\begin{itemize}
\item[] Draw a sample $s^*$ of size $n$ from the population $\U^*$ defined in phase 1, using a high entropy sampling design $P^*$ with first order inclusion probabilities  $\pi_k^*=nx_k^*/ \sum_{j=1}^N x_k^*$.
    \end{itemize}
\end{itemize}
Note that the sampling design $P^*$ used in Phase 2 does not necessarily coincide with the sampling design $P$ used to select the sample $s$ from $\U$, but the resampling inclusion probabilities $\pi_i^*$s have the same structure of the original ones.

This resampling scheme was first considered in \cite{pfefsver04}, in a different framework.
In principle, it is based on a simple idea: Phase 1 mimics the generation process of the finite population from the superpopulation, and
Phase 2 mimics the selection of the sample
from the finite population.
This is sketched in the scheme below.\\

\begin{figure}[H]
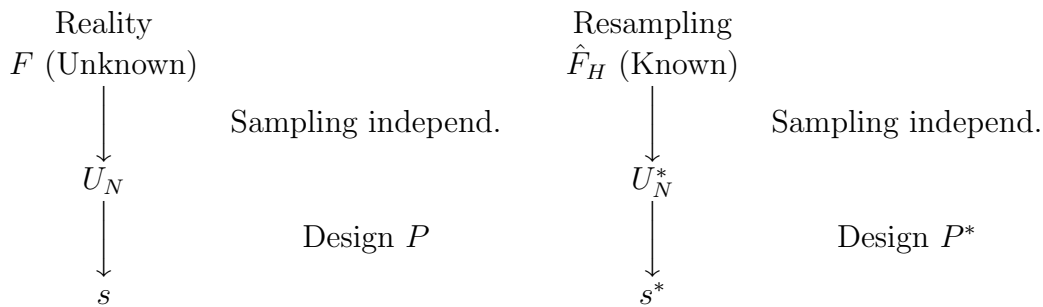

\centering\parbox{\linewidth}{

\begin{tabular}{cccccc}
&Reality & & & Resampling & \\
&$F$ (Unknown) &  &  & $\hat{F}_H$ (Known) &  \\
&$\bigg\downarrow$ & Sampling independ. &  & $\bigg\downarrow$ & {Sampling independ.} \\
&$U_N$ &  &  & $U_N^*$ &  \\
&$\bigg\downarrow$ & Design $P$ &  & $\bigg\downarrow$ & Design $P^*$ \\
&$s$ &  &  & $s^*$ &
\end{tabular}
\caption{Mimicking scheme}
}
\end{figure}
Define now $N_i^*$ as the number of the predicted population units equal to unit $i$ of the sample $s$, and
let $\MP^*$ be the probability distribution of the predicted population generating process. It is easy to see that, given $s$, $\Y$, $\T$, the r.v.s $(N_i^*,\ i \in s)$ possesses a multinomial distribution with:

\begin{align}
&\E_{\MP^*}[N_i^*\vert \D ,\Y , \T]=N\left( D_i\pi_i^{-1}/ \sum_{j=1}^N D_j\pi_j^{-1} \right)\\
&\mathbb{V}_{\MP^*}[N_i^*\vert \D ,\Y , \T]=N\left( D_i\pi_i^{-1}/ \sum_{j=1}^N D_j\pi_j^{-1} \right)\left(1- D_i\pi_i^{-1}/ \sum_{j=1}^N D_j\pi_j^{-1} \right)\\
&\mathbb{C}_{\MP^*}[N_i^*,N_j^*\vert \D ,\Y , \T ]=-ND_iD_j\pi_i^{-1}\pi_j^{-1}/\left( \sum_{k=1}^N D_k\pi_k^{-1} \right)^2, \ j\neq i
\end{align}
The d.f. of the predicted population can be written as:
\begin{equation}
F_N^*(y)=\frac{1}{N}\sum_{i=1}^NI_{(y_i^*\leq y)}=\sum_{i=1}^N D_i \frac{N_i^*}{N}I_{(y_i \leq y)}.
\end{equation}
Consider next the H\'{a}jek  estimator (based on resampled data) of $F_N^*$, which is equal to
\begin{equation}\label{ResampledHajek}
F_H^*(y)=\dfrac{\sum_{i=1}^N\frac{D_i^*}{\pi_i^*}I_{(y_i^*\leq y)}}{\sum_{i=1}^N\frac{D_i^*}{\pi_i^*}}.
\end{equation}

In the sequel, it is shown that the asymptotic distribution of the resampled process:
\begin{equation}\label{RP}
W^{H*}(y)=\sqrt{n}(\widehat{F}_H^*(y)-\widehat{F}_H(y)), \ y\in \mathbb{R}, \ N\geq 1.
\end{equation}
coincides with the asymptotic distribution of $W^{H}$ given in Proposition $\ref{main}$.

\begin{prop}\label{main*}
Suppose the sampling design $P$ and the resampling design $P^*$ both satisfy assumptions $H1-H6$.
The following claims hold.

\noindent {\rm{Claim}} $1$
Conditionally on $\Y, \T, \D$, $N^*_i$s,
the sequence $(W^{H*}_n(y)=\sqrt{n}(\hat{F}_H^*(y)-F_N^*(y)),\ y\in\mathbb{R}, \ N\geq 1)$ converges weakly, in $D[-\infty,\infty]$ equipped with the Skorokhod topology, to a Gaussian Process $\widetilde{W}_1^*$ with zero mean function and covariance function given by $(\ref{C1})$.
The convergence holds for almost all $y_i$s, $t_{ij}$s, for a set of $\D$s of $P$-probability tending to $1$, and for a set
of $N^*_i$s of $\MP^*$-probability tending to $1$.

\noindent {\rm{Claim}} $2$ Conditionally on $\Y, \T, \D$, the
sequence of random functions $(W^{H*}_n(y)=\sqrt{n}(\hat{F}_H^*(y)-F_N^*(y)),\ y\in\mathbb{R}, \ N\geq 1)$ converges weakly, in $D[-\infty,\infty]$ equipped with the Skorokhod topology, to a Gaussian Process $W_1^*$ with zero mean function and covariance function given by (\ref{C1}).
The convergence holds for almost all $y_i$s and $t_{ij}$s, and for a set of $\D$s of $P$-probability tending to $1$.

\noindent {\rm{Claim}} $3$
The two sequences $(W_n^{H*}(y),\ y \in \mathbb{R})$ and $(W_N^*(y),\ y \in \mathbb{R})$ are asymptotically independent.
Moreover, the following statements hold true.
\begin{itemize}
\item[R1] The whole process $(W^{H*}(y),\ y \in \mathbb{R})$ converges weakly in $D[-\infty,\infty]$ endowed with the Skorokhod topology, to a Gaussian process $W^*$ with zero mean function and covariance kernel given by $( \ref{Kernel} )$.
\item[R2] If $\theta(\cdot)$ is continuously Hadamard differentiable at $F$, then $(\sqrt{n}(\theta(\widehat{F}_H^*)-\theta(\hat{F}_H)), \ N \geq 1)$ converges weakly to $\theta'_F(W^*)$, as $N$ increases.
\end{itemize}
In both R1, R2 the convergence hold for almost all $y_i$s and $t_{ij}$s, and for a set of $\D$s of $P$-probability tending to $1$ and
$N$ increases.
\end{prop}

Proposition $\ref{main*}$ shows  that the resampled process
possesses the same limiting behavior as the original limiting process $W$ considered in Proposition \ref{main}.
In the spirit of \cite{r4}, it provides a full asymptotic justification of the resampling procedure considered in the present section.

\begin{rem}
{\rm{For a better understanding of why the resampling scheme introduced so far works, reconsider the \cite{holmberg98} scheme mentioned before,
which is a popular resampling scheme used in finite populations sampling.
For each unit $i$ in the sample $s$, let $R_i = \pi_i^{-1} -  \lfloor \pi_i^{-1} \rfloor$, and consider independent
Bernoulli r.v.s $\epsilon_i$s with $Pr ( \epsilon_i =1  \vert \D , \Y , \, \T ) = R_i$.
Let further $N^*_i = \lfloor \pi_i^{-1} \rfloor + \epsilon_i$. Even if
\begin{eqnarray}
\sum_{i=1}^{N} N^*_i \neq N \nonumber
\end{eqnarray}
it is shown in \cite{contmarmec} that a result similar to Proposition $\ref{main*} $ still holds. In other words the Holmberg scheme is able to recover the limit distribution of the process $W^H_n$. Clearly this is not enough. In our situation we have to take into account the superpopulation randomness (the process $W_N$ that converges to a Brownian Bridge), but the resampled version of $W_N$ under the Holmberg scheme, that is $\sqrt{n}
( F^{*}_{N^{*}} (y) - \widehat{F}_H (y))$, does not converge to a Brownian bridge. To show this, it is enough to observe first that adding and removing the quantity
\begin{eqnarray}
\frac{\displaystyle\sum_{i=1}^{N} \pi_i^{-1}
D_i I_{(y_{i} \leq y)}}{\displaystyle\sum_{i=1}^{N} \left ( \lfloor \pi_i^{-1} \rfloor + \epsilon_i \right )D_i} \nonumber
\end{eqnarray}
we have that
\begin{eqnarray}
\sqrt{n}
( F^{*}_{N^{*}} (y) - \widehat{F}_H (y) ) = A (y) + B(y) \label{scomp_1}
\end{eqnarray}
where
\begin{eqnarray}
A(y) & = & \sqrt{n} \left ( \frac{\displaystyle\sum_{i=1}^{N} \left ( \lfloor \pi_i^{-1} \rfloor + \epsilon_i \right )
D_i I_{(y_{i} \leq y)}}{\displaystyle\sum_{i=1}^{N} \left ( \lfloor \pi_i^{-1} \rfloor + \epsilon_i \right )D_i} -
\frac{\displaystyle\sum_{i=1}^{N} \pi_i^{-1}
D_i I_{(y_{i} \leq y)}}{\displaystyle\sum_{i=1}^{N} \left ( \lfloor \pi_i^{-1} \rfloor + \epsilon_i \right )D_i}
\right ) \nonumber \\
B(y) & = & \sqrt{n} \left ( \frac{\displaystyle\sum_{i=1}^{N} \pi_i^{-1}
D_i I_{(y_{i} \leq y)}}{\displaystyle\sum_{i=1}^{N} \left ( \lfloor \pi_i^{-1} \rfloor + \epsilon_i \right )D_i} -
\frac{\displaystyle\sum_{i=1}^{N} \pi_i^{-1}
D_i I_{(y_{i} \leq y)}}{\displaystyle\sum_{i=1}^{N} \pi_i^{-1}D_i}
\right ) . \nonumber
\end{eqnarray}
Conditionally on  $ \D , \Y , \, \T$, the variance of $\epsilon_i$ is $R_i (1- R_i) \leq 1/4$. Taking into account Lemma $\ref{lemma1}$,
and observing that
\begin{eqnarray}
&\E\left[\displaystyle\sum_{i=1}^{N} \left ( \lfloor \pi_i^{-1} \rfloor + \epsilon_i \right )D_i\right]=
\E_{P}\left[\E_{\MP^*}\left[\displaystyle\sum_{i=1}^{N} \left ( \lfloor \pi_i^{-1} \rfloor + \epsilon_i \right )D_i  \right] \right]=\\
&\E_{P}\left[\displaystyle\sum_{i=1}^{N} \left ( \lfloor \pi_i^{-1} \rfloor + R_i \right )D_i \right]=
\E_{P}\left[\displaystyle\sum_{i=1}^{N} \left ( \pi_i^{-1} \right)D_i \right]=N\footnotemark
\end{eqnarray}
\footnotetext{The symbol $\E_{\MP^*}[\cdot]$ defines the expected value where the only variability is due to the pseudo-population randomness}

this shows that the limiting distribution of $A(y)$ coincides with the limiting distribution of
\begin{eqnarray}
\frac{\sqrt{n}}{N} \sum_{i=1}^{N} ( \epsilon_i - R_i ) D_i I_{(y_i \leq y)} . \nonumber
\end{eqnarray}
In a similar way, it can be shown that the limiting distribution of $B(y)$ coincides with the limiting distribution of
\begin{eqnarray}
- \frac{\sqrt{n}}{N} \sum_{i=1}^{N} ( \epsilon_i - R_i ) D_i F_N (y)  \nonumber
\end{eqnarray}
and hence the limiting distribution of $( \ref{scomp_1} )$ coincides with the limiting distribution of
\begin{eqnarray}
C(y) = \frac{\sqrt{n}}{N} \sum_{i=1}^{N} ( \epsilon_i - R_i ) D_i \left ( I_{(y_i \leq y)} - F_N (y) \right ) . \nonumber
\end{eqnarray}
The arguments of  Proposition $\ref{main}$ can be used to show that, conditionally on $\Y$, $\T$, $C(y)$ converges to a Gaussian process for
almost all $y_i$s, $t_{ij}$s, and for a set of $\D$s of $P$-probability tending to $1$. To show that $C(y)$ does not tend to a Brownian bridge, it is
sufficient to show that the asymptotic variance of $C(y)$ is not $F(y) (1- F(y))$. Since the conditional expectation of $\epsilon_1 - R_i$ is zero,
we have
\begin{eqnarray}
\mathbb{V} ( C(y) \vert \Y , \T ) & = & \frac{n}{N^{2}} \sum_{i=1}^{N} \E \left [
R_i (1- R_i ) D_i \vert  \Y , \T \right ] \left ( I_{(y_{i} \leq y)} - F_N (y) \right )^2  \nonumber \\
\, & = & \frac{n}{N^{2}} \sum_{i=1}^{N} R_i (1- R_i ) \pi_i^{-1 } \left ( I_{(y_{i} \leq y)} - F_N (y) \right )^2  \nonumber \\
\, & \rightarrow & \E_{\MP} [X_1 ] \, \E_{\MP} \left [ R_1 (1- R_1) X_1^{-1} \left ( I_{(Y \leq y)} - F (y) \right )^2 \right ] \nonumber \\
\, & \neq & F(y) (1- F(y)) = \E_{\MP} \left [  \left ( I_{(Y \leq y)} - F (y) \right )^2 \right ] . \nonumber
\end{eqnarray}
The failure of Holmberg scheme is a consequence of a simple fact: the scheme itself cannot recover the generation process of the finite
population from the superpopulation. }}
\end{rem}

\begin{rem}
\rm{In opposition to the Holmberg scheme that fails because of its inability of recovering the superpopulation model contribution to the limiting distribution of the process \eqref{Proc}, the Efron's Bootstrap fails because it is not able to recover the contribution of the sampling design to the limiting distribution of the sequence \eqref{Proc}. In fact, let $s^*$ be a sample of $n$ $i.i.d.$ observations obtained by sampling with replacement from the
consistent estimator $\widehat{F}_H$ of the superpopulation distribution function $F$. Consider now the resampling process
\begin{equation}\label{boot_proc}
\sqrt{n}\left(\widehat{F}_n^*-\widehat{F}_H \right)
\end{equation}
\noindent Following the same approach of \cite{r4}, it easy to see that the process \eqref{boot_proc}, conditionally on the original sample and finite population, behaves as a Brownian bridge when increasing $n$ and $N$. The latter consideration implies that the classic Bootstrap procedure is able to recover only the contribution of the superpopulation model to the limiting distribution of the process \eqref{Proc}, completely ignoring the variability due to the sampling design. }
\end{rem}
\section{Resampling procedure: Monte Carlo algorithm}\label{s5b}

Clearly, resampling is performed by resorting to Monte Carlo simulations and this is a computer-intensive procedure. Thus, due to the factorial growth of the cardinality of the space of the boostrap samples, recovering the true asymptotic (resampling) distribution is practically infeasible.
To avoid this problem, we want to approximate the true asymptotic distribution of the H\'{a}jek  estimator with the simulated resampling distribution. This procedure will be now clarified. For the sake o simplicity we assume $\theta(\cdot)$ to be real-valued, that is considering scalar parameters of the superpopulation.
\begin{itemize}
\item[1.] Generate $M$ independent bootstrap samples of size $n$ on the basis of the two-phase resampling procedure described above.
\item[2.] For each bootstrap sample, compute the corresponding H\'{a}jek estimator (\ref{ResampledHajek}), denoted by $\hat{F}_{H,m}^*, \ m=1,2,\ldots,M$.
\item[3.] Compute the corresponding estimates of $\theta(\cdot)$:
$$ \hat{\theta}^*_m=\theta(\hat{F}_{H,m}^*), \ m=1,2,\ldots,M. $$
\item[4.] Compute the $M$ quantities
\begin{equation}\label{Znm}
Z_{n,m}^*=\sqrt{n}(\hat{\theta}^*_m-\theta(\hat{F}_{H}^*)), \ m=1,2,\ldots,M.
\end{equation}
\item[5.] Compute the variance of (\ref{Znm}):
\begin{equation}
\hat{S}^{2*}=\dfrac{1}{M-1}\sum_{m=1}^M(Z_{n,m}^*-\bar{Z}_M^*)^2=\dfrac{n}{M-1}\sum_{m=1}^M(\hat{\theta}^*_m-\bar{\theta}_M^*)
\end{equation}
where
$$\bar{Z}_M^*=\dfrac{1}{M}\sum_{m=1}^M Z_{n,m}^*, \ \bar{\theta}_M^*=\dfrac{1}{M}\sum_{m=1}^M \hat{\theta}^*_m.  $$
Denote further by
\begin{equation}\label{ECDF}
\hat{R}^*_{n,M}(z)=\frac{1}{N}\sum_{m=1}^MI_{(Z_{n,m}^*\leq z)}, \ z\in \mathbb{R}
\end{equation}
\noindent the empirical distribution function of $Z_{n,m}^*$.
\end{itemize}
The empirical distribution (\ref{ECDF}) is essentially an approximation of the resampling distribution of $\sqrt{N}(\theta(\hat{F}_N^*(y))-\theta(\hat{F}_H(y)))$.
Next proposition establishes convergence of the empirical distribution (\ref{ECDF}) to the actual asymptotic distribution of the resampled process.
\begin{prop}\label{Last-prop}
Suppose the assumptions $H1-H6$ are fulfilled, let $\sigma^2_{\theta}=\mathbb{V}_{\MP}(\theta(F))$, and let $\Phi_{0,\sigma^2_{\theta}}$ be a normal distribution function with expectation $0$ and variance $\sigma^2_{\theta}$. The following result holds:
\begin{equation}\label{LP-1}
\sup_z\lvert \hat{R}^*_{n,M}(z) - \Phi_{0,\sigma^2_{\theta}}(z) \rvert \xrightarrow{a.s.-\MP^*} 0,\ \text{as } M,N \text{ go to infinity.}
\end{equation}
The convergence holds for a almost all $y_i$s and $t_{ij}$s, for a set of $\D$s of $P-$probability tending to $1$, and is in probability w.r.t. $P^*$.\\
If, in addition, $\E_{\MP^*}[Z_{n,m}^{2*}]<\infty$,  the sample variance $S^{2*}$ of $(Z_{n,m}^*, \ m=1,2,\ldots,M)$ is a consistent estimator of
$\sigma^2_{\theta}$, as $M,N$ tend to infinity.
\end{prop}.

\section{Applications}\label{s6}
\subsection{Confidence intervals for quantiles}
The aim of this Section is to test the performance of our resampling procedure when dealing with confidence intervals for quantiles.
The (superpopulation) quantile function is
\begin{equation}
Q(p)=\inf\{y \in \mathbb{R} :\ F(y)\geq p  \}=F^{-1}(p), \ \text{with } 0<p<1 ,
\end{equation}
{\em i.e.} $Q(\cdot)$ is the left-continuous inverse function of $F(\cdot)$.
Let now focus on the real-valued functional $\theta_p(\cdot): D[-\infty,+\infty]\to \mathbb{R}$ that brings the distribution function $F$ in its quantile of order $p$ (i.e. $\theta_p(F)=F^{-1}(p)=Q(p)$).
In \cite{r9} Lemma 21.3 it is shown that, if $F$ is differentiable at point $q_p$ (such that $F(q_p)=p$), with $F'(q_p)=p>0$, then $\theta_p(F) = Q(p)$ is Hadamard-differentiable at $F$.
As a consequence, all the results obtained in Sections \ref{s3}-\ref{s5} are valid, and
\begin{equation}\label{CI}
[\hat{L}^j,\hat{U}^j]=\left[\theta_p(\hat{F}_H)+z_{\frac{\alpha}{2}}\dfrac{S^*}{\sqrt{n}},\theta_p(\hat{F}_H)+z_{1-\frac{\alpha}{2}}\dfrac{S^*}{\sqrt{n}} \right]
\end{equation}
is a confidence interval for $Q(p)$ of asymptotic size $1-\alpha$, where $z_{\alpha}$ is the quantile of order $\alpha$ of a standard normal distribution.\\
In order to test the performance of our resampling procedure we conduct a small simulation study.
For our simulations we assume the same superpopulation model as in \cite{r10}, {\em i.e.}
\begin{equation}\label{Pop}
Y=(\beta_0+\beta_1X^{1.2}+\sigma \epsilon)^2+c
\end{equation}
where $X\sim |N(0,7)|$, $\epsilon\sim N(0,1)$, $\beta_0=12.5$, $\beta_1=3$, $\sigma=15$ and $c=4000$.
Parameters in (\ref{Pop}) are chosen in order to have a distribution of the character of interest similar to a log-normal distribution.
In addition, a design variable $X=Y^{0.2}W$, with $W\sim \log N(0,0.125)$, is considered, and inclusion probabilities are taken proportional to $X$ values.
In both sampling and resampling procedures we consider Pareto design.
In the simulation study we have investigated the behavior of our proposal in two situations: large sampling fraction ($f=1/3$), and small sampling fraction ($f=1/10$), with different sample sizes ($n=50, \, 150$). For each sampling fraction and sample size we have generated $1000$ finite populations and for each sample selected from these populations, $M=1000$ bootstrap samples are drawn.
Using our resampling scheme, confidence intervals for quantiles of order $p=0.10, 0.25, 0.5, 0.75, 0.9$ have been computed according to formula (\ref{CI}), with a confidence level of $95\%$. In order to test the performance of our procedure, the following indicators have been computed.
\begin{itemize}
\item[1.] Estimated Coverage Probability
\begin{equation}\label{CP}
CP=\displaystyle \dfrac{1}{M}\sum_{j=1}^M I(\hat{L}^j\leq \hat{q}_p\leq \hat{U}^j) .
\end{equation}
\item[2.] Estimated Left and Right Errors
\begin{align}
\label{LE}&LE=\dfrac{1}{M}\sum_{j=1}^M I(\hat{L}^j> \hat{q}_p);\\
\label{RE}&RE=\dfrac{1}{M}\sum_{j=1}^M I(\hat{U}^j< \hat{q}_p) .
\end{align}
\item[3.] Average Length
\begin{equation} \label{AVL}
AL=\dfrac{1}{M}\sum_{j=1}^M \left(\hat{U}^j-\hat{L}^j\right).
\end{equation}
\end{itemize}
In $(\ref{CP})$ - $(\ref{RE})$ the quantity $\hat{q}_p$ is the empirical $p-$quantile obtained simulating $1000000$ values from the model (\ref{Pop}).\\
Next tables show the estimated quantities $(\ref{CP})$ - $(\ref{AVL})$ in different situations.\\

\begin{table}[h]
\centering\parbox{0.6\linewidth}{

\begin{tabular}{ |l|c|c|c|c|c| }
  \hline
  \multicolumn{6}{|c|}{$\mathbf{f=1/10}$,\ $\mathbf{1-\boldsymbol{\alpha}=0.95}$,\ $\mathbf{n=50}$} \\
  \hline
  \backslashbox{}{p} & $0.10$ & $0.25$ & $0.50$ & $0.75$ & $0.90$ \\
\hline
CP & $0.948$ & $0.943$ & $0.939$ & $0.931$ & $0.921$ \\
LE & $0.027$ & $0.02$ & $0.026$ & $0.017$ & $0.014$ \\
RE & $0.025$ & $0.037$ & $0.035$ & $0.052$ & $0.065$ \\
AL & $302.751$ & $663.436$ & $1284.412$ & $2588.063$ & $5195.219$ \\
  \hline
\end{tabular}
\caption{Results with a finite population of $N=500$ units, a true confidence level of $95\%$ and sample size $n=50$.}
\label{table:1}
}

\end{table}

\begin{table}[h]
\centering\parbox{0.6\linewidth}{

\begin{tabular}{ |l|c|c|c|c|c| }
  \hline
  \multicolumn{6}{|c|}{$\mathbf{f=1/3}$,\ $\mathbf{1-\boldsymbol{\alpha}=0.95}$,\ $\mathbf{n=50}$} \\
  \hline
\backslashbox{}{p} & $0.10$ & $0.25$ & $0.50$ & $0.75$ & $0.90$ \\
\hline
CP & $0.949$ & $0.944$ & $0.933$ & $0.929$ & $0.925$ \\
LE & $0.03$ & $0.022$ & $0.017$ & $0.023$ & $0.018$ \\
RE & $0.021$ & $0.034$ & $0.05$ & $0.048$ & $0.057$ \\
AL & $301.26$ & $654.764$ & $1274.644$ & $2583.419$ & $5143.599$ \\
  \hline
\end{tabular}
\caption{Results with a finite population of $N=150$ units, a true confidence level of $95\%$ and sample size $n=50$.}
\label{table:2}}

\end{table}

\begin{table}[H]
\centering\parbox{0.6\linewidth}{

\begin{tabular}{ |l|c|c|c|c|c| }
  \hline
  \multicolumn{6}{|c|}{$\mathbf{f=1/10}$,\ $\mathbf{1-\boldsymbol{\alpha}=0.95}$,\ $\mathbf{n=150}$} \\
  \hline
  \backslashbox{}{p} & $0.10$ & $0.25$ & $0.50$ & $0.75$ & $0.90$ \\
\hline
CP & $0.939$ & $0.957$ & $0.949$ & $0.941$ & $0.941$ \\
LE & $0.02$ & $0.016$ & $0.019$ & $0.014$ & $0.009$ \\
RE & $0.041$ & $0.027$ & $0.032$ & $0.045$ & $0.05$ \\
AL & $149.996$ & $363.505$ & $700.394$ & $1438.521$ & $2874.784$ \\
  \hline
\end{tabular}
\caption{Results with a finite population of $N=1500$ units, a true confidence level of $95\%$ and sample size $n=150$.}
\label{table:3}
}

\end{table}

\begin{table}[H]
\centering\parbox{0.6\linewidth}{

\begin{tabular}{ |l|c|c|c|c|c| }
  \hline
  \multicolumn{6}{|c|}{$\mathbf{f=1/3}$,\ $\mathbf{1-\boldsymbol{\alpha}=0.95}$,\ $\mathbf{n=150}$} \\
  \hline
\backslashbox{}{p} & $0.10$ & $0.25$ & $0.50$ & $0.75$ & $0.90$ \\
\hline
CP & $0.94$ & $0.94$ & $0.937$ & $0.943$ & $0.91$ \\
LE & $0.031$ & $0.019$ & $0.025$ & $0.015$ & $0.02$ \\
RE & $0.029$ & $0.041$ & $0.038$ & $0.042$ & $0.07$ \\
AL & $154.518$ & $353.819$ & $698.809$ & $1418.626$ & $2893.49$ \\
  \hline
\end{tabular}
\caption{Results with a finite population of $N=450$ units, a true confidence level of $95\%$ and sample size $n=150$.}
\label{table:4}}

\end{table}

From tables \ref{table:1}-\ref{table:4}, it is seen that the estimated coverage probability is close to the nominal confidence level in both situations
of large and small sampling fractions, as well as for the different considered sample sizes.
As it can be expected, a lower performance is obtained when a confidence interval for extreme quantile of order $p=0.90$ is considered. In this case we have the worst estimated coverage probability in every simulated scenario, and also the most unbalanced tail errors. Generally, problem of estimating extreme quantiles is a hard problem. Going further, the considered population is highly positive skewed, thus the estimation of quantities in the right tail of the population distribution can be difficult.

\subsection{Testing for conditional independence}
The goal of this paragraph is to perform an independence test for two interest characters, conditionally on discrete design variables $T_js$. For the sake of simplicity we will consider a single design variable $T$, thus we are considering a test of the form
$$\begin{cases}
H_0: &H(x,y|T)=F(x|T)G(y|T)\\
H_1: &H(x,y|T)\neq F(x|T)G(y|T)\\
\end{cases}$$
To achieve this purpose, the general measure of monotone dependence, proposed in \cite{r12} is extended to the present case.
Given two continuous variables $X,Y$, let $F(x)$ and $G(y)$ be their marginal distributions and $H(x,y)$ the joint distribution of the bivariate variable $(X,Y)$.
A general measure of the monotone dependence $\gamma_g$ between $X$ and $Y$, is a real-valued functional $\gamma_g$ of the bivariate distribution $H(x,y)$ defined as follows
\begin{equation}\label{def:gamma}
\gamma_g=\int_{\mathbb{R}^2}g(|F(x|T)+G(y|T)-1|)-g(|F(x|T)-G(y|T)|) \ dH(x,y|T),
\end{equation}
where $g:[0,1]\to\mathbb{R}$ is a strictly increasing, continuous and convex function, such that $g(0)=0$ with continuous first derivative. Under the null hypothesis of independence the latter quantity is equal to zero.\\
The basic idea is to estimate the quantity $\gamma_g$ with a plug-in approach, substituting the distributions functions
\begin{equation}\label{def:samplegammacond}
\hat{\gamma}_{g,H|T}=\dfrac{\displaystyle\sum_{i\in s}\frac{1}{\pi_i} \left( g(|F(x_i|T_i)+G(y_i|T_i)-1|)-g(|F(x_i|T_i)-G(y_i|T_i)|)\right)}{\displaystyle\sum_{i \in s} \frac{1}{\pi_i}}.
\end{equation}
Before analyzing the Hadamard-differentiability of \eqref{def:gamma}, we stress that our results are given for the univariate case, but they can be simply generalized to the multivariate case.
To show the Hadamard-differentiabilty of the considered functional, it is sufficient to use the same arguments as the proof of Theorem 4.1. in \cite{r12} and then use result (4) in \cite{r23}.\\
Before illustrating our simulation study, it is important to stress a couple of remarks.\\
First of all, the design variable is supposed discrete for the sake of simplicity. In fact, estimating conditional d.f. when the conditioning variable is discrete, does not involve different estimation techniques, but only focusing on a subgroup of the population than the whole population. Allowing the conditioning variable being continuous implies more complex estimator of the distribution function (like kernel estimators) that fall outside the spirit of our paper.
The second remark is about the resampling procedure. In order to perform a test with resampling techniques, it is necessary to resample under the null hypothesis, thus in our case we need to resample under the hypothesis of conditional independence of the the two interest characters $X,Y$. To this purpose the pseudo-population generation phase of our resampling technique has been modified as follows.
According to the previous notation $X,Y$ are variables of interest and $T$ takes values $T^1,\ldots,T^k$. In addition, let $s$ be a sample of units selected from a $N-sized$ finite population $\mathcal{U}_n$ with a $\pi p s$ sampling design $P$, where the inclusion probabilities $\pi_i \propto T_j$. Define
$s_j=\{i \in s | t_i=T^j\},\ j=1,\ldots,k$, the set of sampled units with T-value equal to $T^j$. Let $n_1,\ldots,n_k$ be the size of $s_1,\ldots,s_k$.
Firstly, a pseudo-population of $N$ values $T_1^*,\ldots,T_N^*$ is generated, where each unit is selected independently with probability
$\pi_i^{-1}/  \sum_{j\in s} \pi_j^{-1}$. Then, for $l=1,\ldots,N$, if $T_l^*=T^j,\ j=1,\ldots,k$ we sample independently from $s_j$, with probability $\pi_i^{-1}/  \sum_{j\in s} \pi_j^{-1}$ a $X$-value $X_l^*$ and a $Y$-value $Y_l^*$.
At the end of this procedure, a pseudo-population $\mathcal{U}^*_N=(X_l^*,Y_l^*,T_l^*,\ l=1,\ldots,N)$ is obtained, where such that $X^*$ and $Y^*$ are independent conditionally on $T^*$. At this point the second phase of the resampling method as shown in section \ref{s5}, can be used.
The considered resampling scheme to recover the distribution of $\sqrt{n}(\hat{\gamma}_{g,H|T}-\gamma_{g|T})$ under the null hypothesis of independence and hence to perform the test.\\
In the sequel, we will focus on a simulation study where the function $g(s)=s^2$. With this choice of $g$, the coefficient $\gamma_g$ become exactly the non-normalized version of the Spearman's rank coefficient $\rho_s$ (for more see \cite{r12}).

For the simulations study we assumed that in the superpopulation there are four strata, indexed by the discrete variable $T \in \{1,2,3,4\}$. For each stratum we have the interest variables $(X,Y)$ distributed as a bivariate normal $N(\mu_T,\Sigma)$ where
\begin{align}
\label{sigma:spear}\Sigma&=\left(\begin{array}{cc}
150^2 & 150\cdot 60\cdot 2\cdot \sin(\frac{\pi}{6}\cdot \rho_s) \\
150\cdot 60\cdot 2\cdot \sin(\frac{\pi}{6}\cdot \rho_s) & 60^2
\end{array}  \right)\\
\mu_T&=\begin{cases}
(800,300)' &\text{if } T=1\\
(900,400)' &\text{if } T=2\\
(1000,500)' &\text{if } T=3\\
(1100,600)' &\text{if } T=4\\
\end{cases}
\end{align}
In addition, each stratum has a weight in the superpopulation equal to $$
\omega_T=\begin{cases}
0.4 &\text{if } T=1\\
0.3 &\text{if } T=2\\
0.2 &\text{if } T=3\\
0.1 &\text{if } T=4\\
\end{cases}.$$
Setting the covariance matrix as in (\ref{sigma:spear}), involve having exactly a Spearman's correlation coefficient between $X,Y$ of $\rho_s$  (for a proof see for instance \cite{r25}), in each group. Thus we have an overall Spearman's correlation coefficient conditionally on $T$ equal to $\rho_s$.
In this setting our test becomes
$$\begin{cases}
H_0: &\rho_{s|T}=0\\
H_1: &\rho_{s|T}\neq 0\\
\end{cases}$$
with a (estimated) region of rejection of the form \{$\lvert\hat{\rho}_{s,H|T}\rvert>c(\alpha)$\}, where $c(\alpha)$ is the $1-\frac{\alpha}{2}$-quantile of the resampling (under null hypothesis) distribution of $\hat{\rho}_{s,H|T}$.\\
As in the previous paragraph, sample sizes $n=50,150$ and sampling fractions $f=1/3,1/10$ are considered.
For each sample size and sampling fractions, $1000$ finite populations have been generated and for each sample selected from these populations, $M=1000$ bootstrap samples have been drawn. In addition, two sampling scenarios have been considered. The first one (CP-PA) where samples are selected according to a Conditional Poisson (CP) sampling design and in resampling procedure a Pareto (PA) design has been used. The second one (PA-PA) where both sampling and resampling are implemented in a Pareto (PA) design.\\
A test of nominal level $\alpha=5\%$ has been performed to evaluate the performance of such procedure, the estimated type $I$ error, the median of estimated  P-value (to limit the influence of the extreme estimated P-values) and the estimated power function, have been computed. The results of our simulation study are summarized below.
\begin{table}[h]
\centering\parbox{0.6\linewidth}{
\begin{tabular}{|l|c|c|}
\hline
Sample size and Sampling fraction  & $\hat{\alpha}$ (CP-PA)&$\hat{\alpha}$ (PA-PA)  \\
\hline
$n=50,\ \ f=0.1$ & $0.053$ & $0.051 $     \\
$n=150,\ f=0.1$ & $0.045$ &  $0.048 $     \\
$n=50,\ \ f=0.3$ & $0.06$ &  $0.051 $    \\
$n=150,\ f=0.3$ & $0.05$&    $0.048 $     \\
\hline
\end{tabular}
\caption{Estimated first type error for different sample size and sampling fractions. Nominal $\alpha=5\% $.}
\label{table:5}}
\end{table}

\begin{figure}[H]
\centering
  \includegraphics[width=\linewidth]{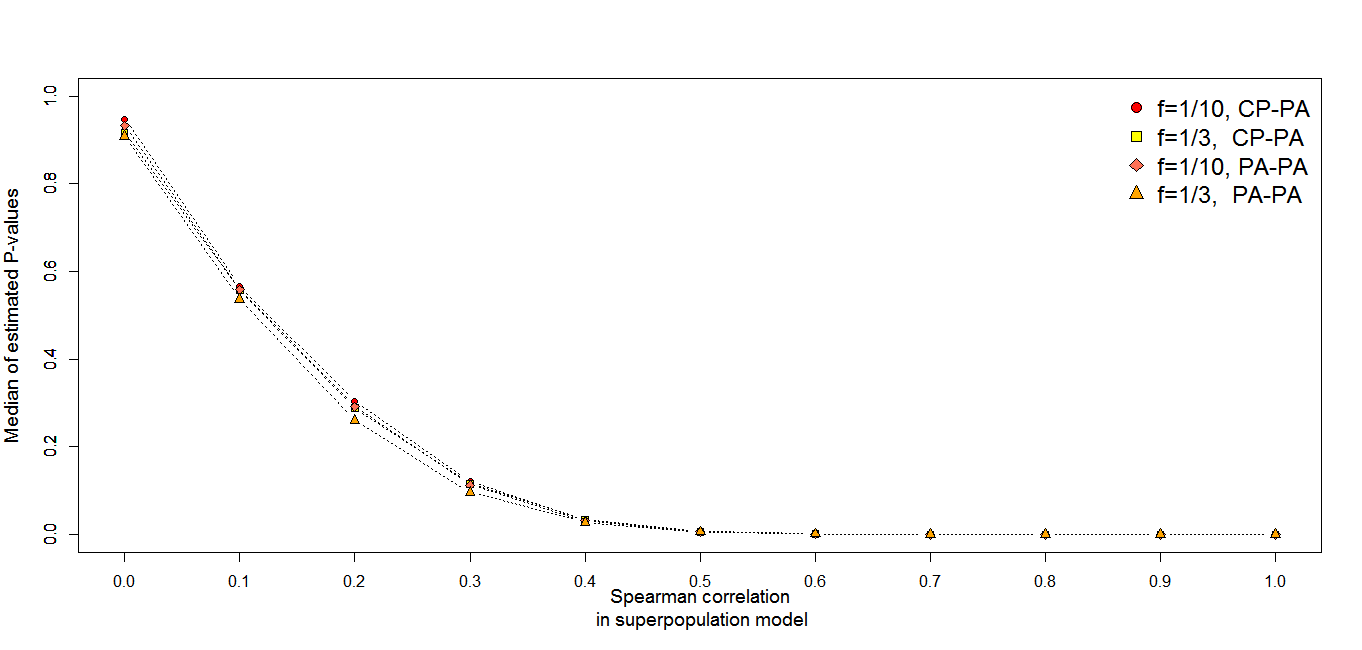}
  \caption{Median of estimated P-values for each level of correlation with $n=50, \ f=1/3,\ f=1/10.$ and considering CP-PA, PA-PA scenarios}
  \label{fig:sub1}
\end{figure}
\begin{figure}[H]
  \centering
  \includegraphics[width=\linewidth]{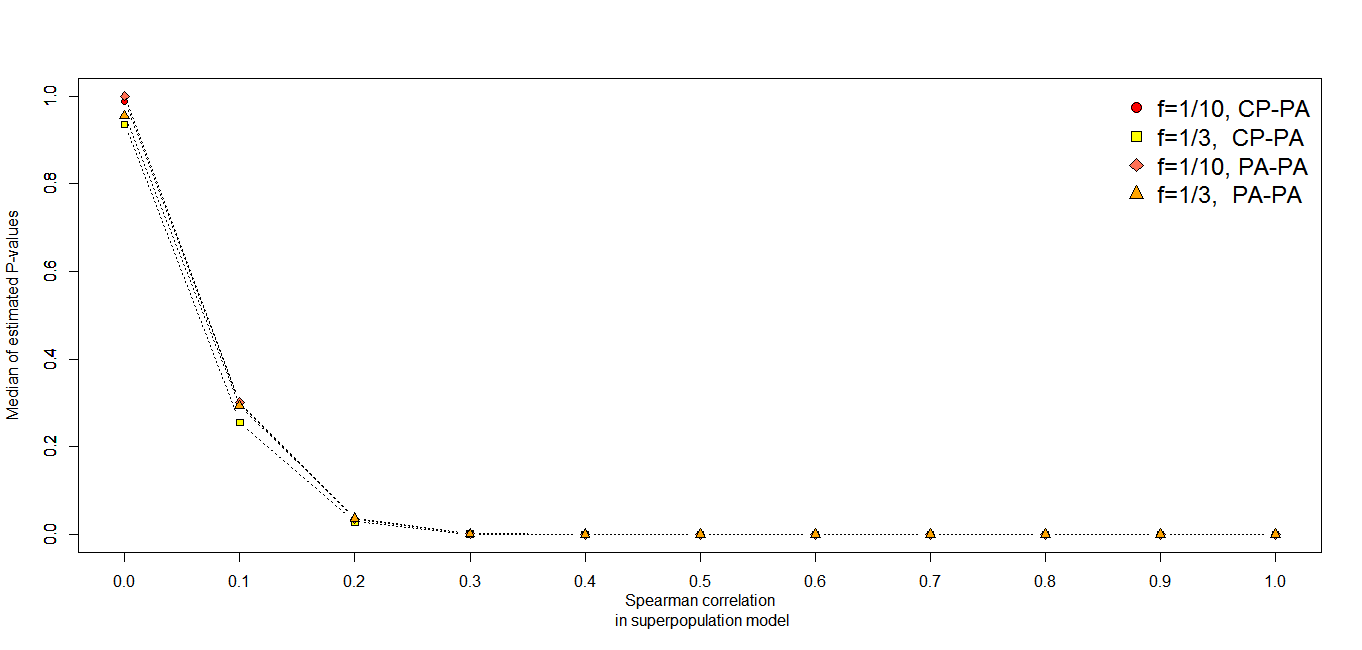}
  \caption{Median of estimated P-values for each level of correlation with $n=50, \ f=1/3,\ f=1/10.$ and considering CP-PA, PA-PA scenarios}
  \label{fig:sub2}
\end{figure}

\begin{figure}[H]
\centering

  \includegraphics[width=1\linewidth]{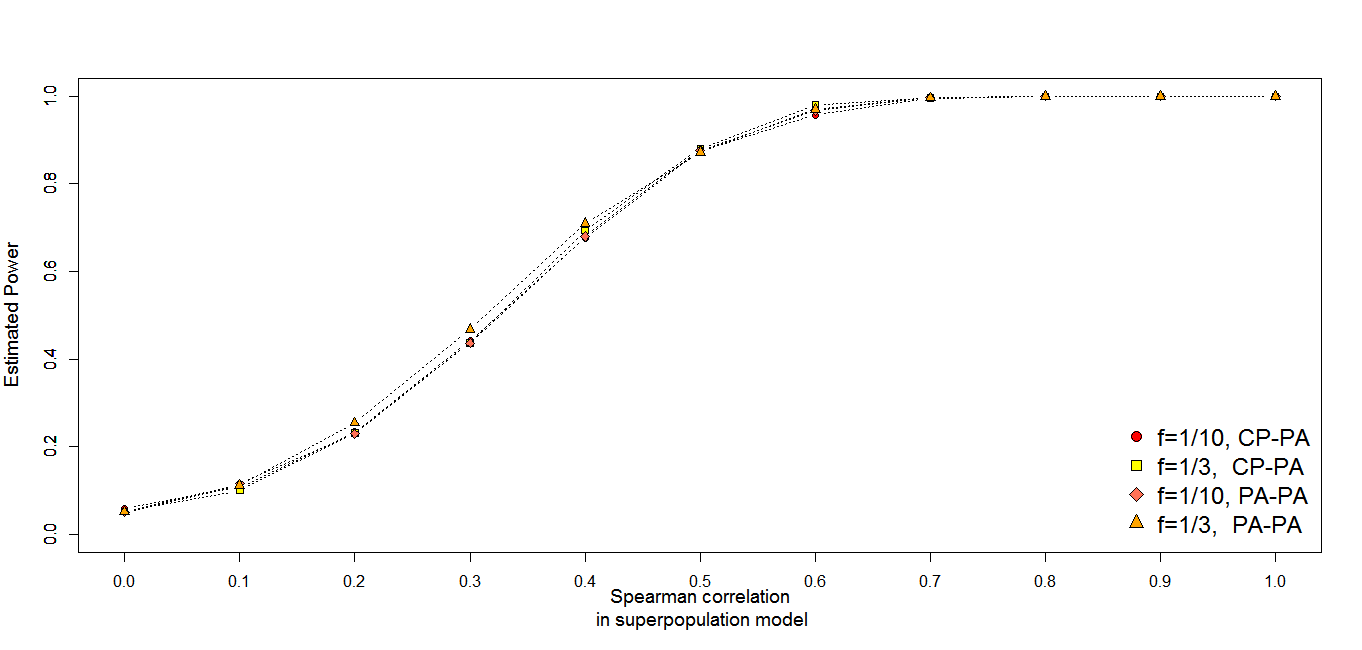}
  \caption{Estimated power function where $n=50, \ f=1/3,\ f=1/10.$ and considering CP-PA, PA-PA scenarios}
  \label{fig:sub3}
\end{figure}
\begin{figure}[H]
  \centering
  \includegraphics[width=1\linewidth]{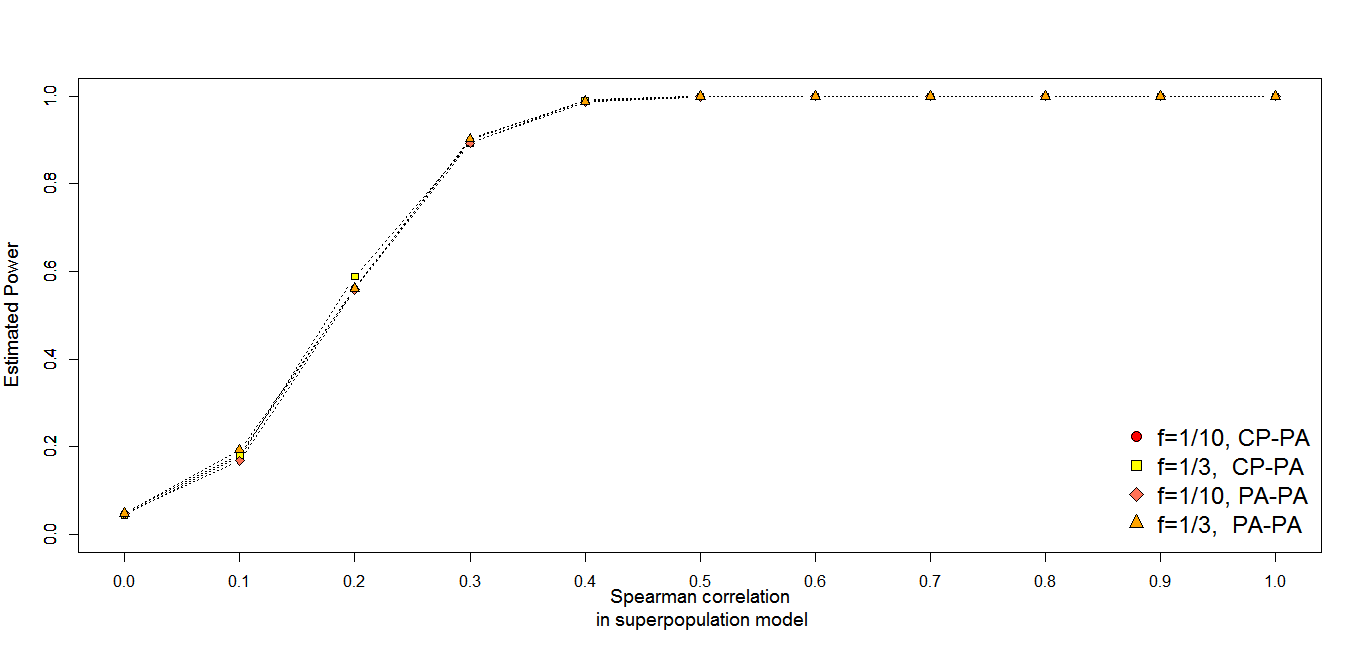}
  \caption{Estimated power function where $n=150, \ f=1/3,\ f=1/10.$ and considering CP-PA, PA-PA scenarios}
  \label{fig:sub4}
\end{figure}

From Table \ref{table:5}, it is immediately seen that our procedure works well in both situations of a small and big sampling fractions. In addition, we can notice that using two different sampling designs for the sampling and resampling stages give results similar to those obtained by using the same sampling design. In both situations the estimated level $\hat{\alpha}$ is very close the nominal level of $5\%$. Another important remark is that PA-PA scenario seems to be more stable respect to the CP-PA one, in the sense that the estimated type $I$ error fluctuates less in PA-PA scenario than in CP-PA.
As far as the estimated P-values and the estimated power function are concerned, it is seen from Figures \ref{fig:sub1}-\ref{fig:sub4} that the differences for the same sample size in the different scenarios are small, but we have generally lower P-values, thus a better performance, when considering $f=1/3$. For a sample size of $n=50$ the median of estimated P-values becomes zero when the Spearman correlation is 0.6. Of course increasing sample size results in a decrease of the Spearman correlation beyond which the median of the estimated P-values is zero. The analysis of estimated power functions involves similar conclusions. In fact, for a sample size $n=50$ the estimated power functions are very similar, but the power is higher in the case of a larger sampling fraction. This result is reasonable. In fact, the larger the sampling fraction the larger the information that the sample carries. In this particular case, a larger sampling fraction allows more easily the sample to reconstruct the correlation structure in the finite population (and in the superpopulation).

\subsection{Testing for marginal independence}
The goal of the present section is to construct a test for the marginal independence of two (continuous) characters of interest $Y$, $Z$, without conditioning on the design variables $T_j$s.
For the sake of simplicity, in the sequel we will consider a single design variable $T$, say.
The general framework is the same of the previous Section, {\em i.e.} from the above results we use the measure of monotone dependence with $g(s)=s^2$ to test
$$\begin{cases}
H_0: &\rho_{s}=0\\
H_1: &\rho_{s}\neq 0\\
\end{cases}.$$
Clearly all the results derived in the previous paragraph are still valid, in particular we have that $\sqrt{n}(\hat{\gamma}_{g,H}- \hat{\gamma}_{g})$ is asymptotically normal with zero mean and a complex variance that depends on the Hadamard derivative of the functional $\gamma_g$.
Although from the theoretical point of view it seems to be an easy problem, from the practical point of view it presents more difficulties than the case analyzed before. In fact, performing a test with resampling requires the ability of sampling under the null hypothesis. In this framework, for each sampling unit we have a
triplet $(y_i , \, z_i , \, t_i)$; thus, a unique sample value of $T$ is  associated to each pair $(y_i , \, z_i)$.
In order to apply our resampling procedure to the testing problem, we have to generate a pseudo-population $Y_i^* , \, Z_i^* , \, Y_i^*$ from the sample values in such a way that $Y^*$ and $Z^*$ are marginally independent (null hypothesis). Independence can be obtained by sampling independently from the marginal (H\'{a}jek) estimators of the distribution functions of $Y$, $Z$.  However, in this way it is not possible to uniquely associate a value $T_i^*$ to each pair ($Y_i^* , \, Z_i^*$). To avoid this problem,
we look at the testing problem as the inverse of an interval confidence problem: an asymptotic confidence interval of size $1-\alpha$ provides an asymptotic test of size $\alpha$.
Of course this way of looking at the problem is simpler but has some limits. One of them is that we cannnot provide estimated p-values, because for
 their computation it is necessary to resample under the null hypothesis.\\
With the previous notation, the following interval
\begin{align}
\left[\hat{\gamma}_{g,H}+z_{\frac{\alpha}{2}}\sqrt{\frac{S^{2*}}{n}},\hat{\gamma}_{g,H}+z_{1-\frac{\alpha}{2}}\sqrt{\frac{S^{2*}}{n}}\right]
\end{align}
is a confidence interval for $\gamma_g$ of asymptotic size $1-\alpha$. The null hypothesis of independence is accepted if 0 lies in the interval,  and rejected otherwise.\\
In our simulation study, $(Y, \, Z)$ is assumed to be a bivariate Marshall-Olkin copula (for more see \cite{r27}, \cite{r28}, \cite{r26}). One of the advantages of the bivariate Marshall-Olkin copula is that it allows a Spearman's correlation coefficient that has an analytic form, that only depends on the parameter of the copula (as for the Gaussian copula used in the previous paragraph), and that takes value in the interval $[0,1]$.
For the simulation study three different sample sizes, $n=50, \, 150, \,250$ have been considered, in both situations of a large ($f=1/3$) and small ($f=1/10$) sampling fractions.
For each sample size and sampling fraction, $1000$ finite populations have been generated, and for each sample selected from these populations, $M=1000$ samples have been drawn.
Samples were selected according to a Conditional Poisson design. As far as the  resampling stage is concerned, a Pareto design was used.
The inclusion probabilities $\pi_i$ have been taken proportional to
$T=f(U) W$, where $U=Y+Z$, $f(u)=u^3/3-0.5u^2+0.10u+0.5$ and $W\sim \log N(0,\sigma^2)$ with $\sigma^2=0.4$ if $f=1/10$ and $\sigma^2=0.08$ if $f=1/3$.
The design variable $T$ possesses  correlation with $Y$ and $Z$ ranging in between $0.4$ and $0.5$, and a broad range of variation of the inclusion probabilities (about $[0.02,0.95]$).
Tests of different sizes $\alpha=0.1, \, 0.05, \, 0.01$ have been performed.
To evaluate the performance of our procedure, estimated type $I$ error probabilities ($\hat{\alpha}$) have been computed,  
as well as estimated power functions for different situations sampling fractions and sample sizes. 
Results are summarized below.
\begin{table}[h]
\centering\parbox{\linewidth}{

\begin{tabular}{ l|c|c||c|c||c|c|| }
  \cline{2-7}
  & \multicolumn{2}{|c||}{$\boldsymbol{\alpha =}\textbf{0.1}$} & \multicolumn{2}{|c||}{$\boldsymbol{\alpha =}\textbf{0.05}$} &\multicolumn{2}{|c||}{$\boldsymbol{\alpha =}\textbf{0.01}$} \\ \cline{2-7}

   & $f=1/10$ & $f=1/3$ & $f=1/10$ & $f=1/3$ & $f=1/10$ & $f=1/3$ \\
\hline
\multicolumn{1}{|c|}{$n=50$} & $0.124$ & $0.116$& $0.074$ & $0.065$  & $0.02$  & $0.017$ \\
\multicolumn{1}{|c|}{$n=150$}& $0.126$ & $0.11$ & $0.064$ & $0.062$  & $0.021$ & $0.012$ \\
\multicolumn{1}{|c|}{$n=250$}& $0.109$ & $0.1 $ & $0.055$ & $0.061$  & $0.012$ & $0.016$ \\
\hline
\end{tabular}
\caption{Estimated  type $I$ error probability $\hat{\alpha}$, for different sample sizes and sampling fractions.}
\label{table:6}
}
\end{table}\\
For the sake of brevity, only graphs of estimated power functions for a nominal level $\alpha=0.05$ are shown.
\begin{figure}[H]

  \centering
  \includegraphics[width=1\linewidth]{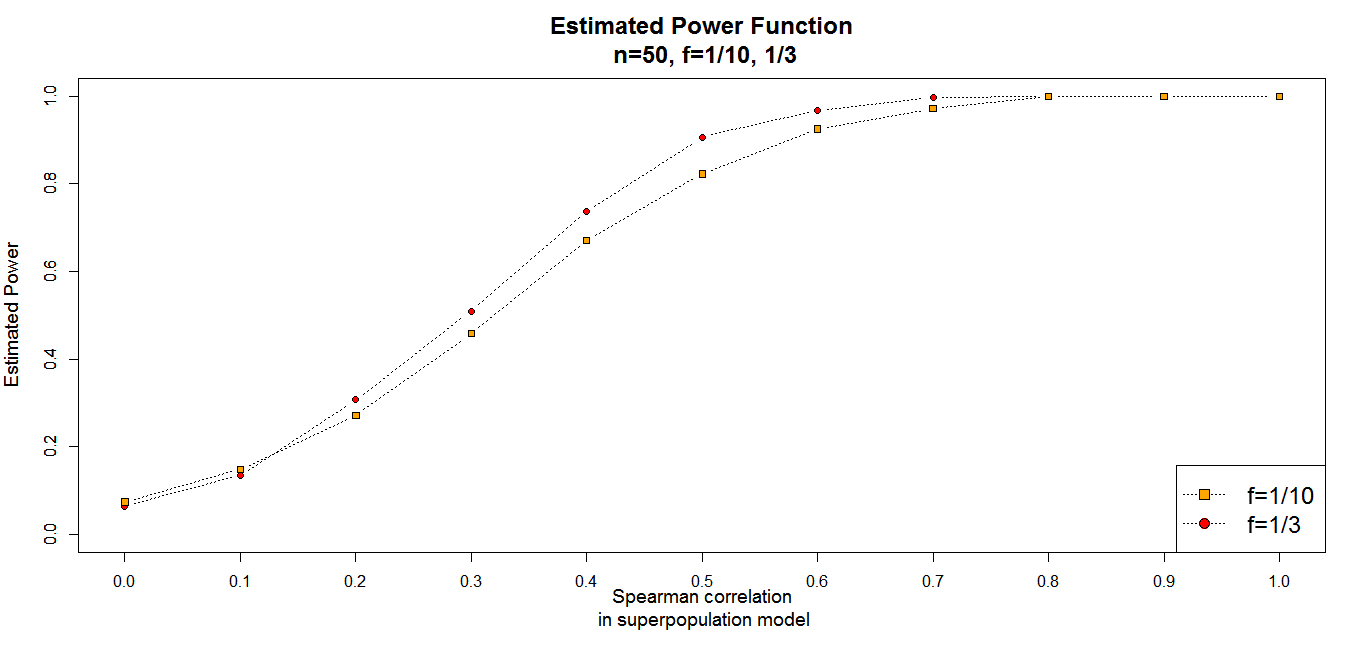}
  \caption{Estimated power function where $n=50, \ f=1/3,\ f=1/10.$}
  \label{fig:sub5}
\end{figure}

\begin{figure}[H]
  \centering
  \includegraphics[width=1\linewidth]{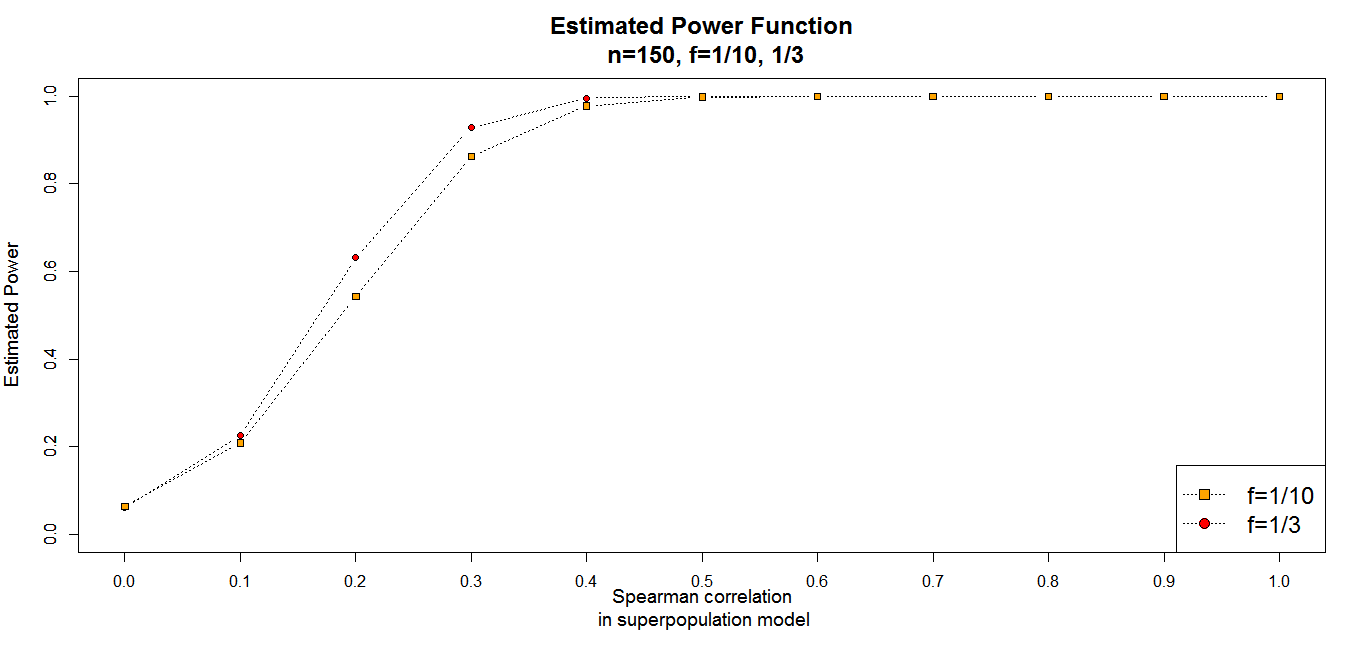}
  \caption{Estimated power function where $n=150, \ f=1/3,\ f=1/10.$}
  \label{fig:sub6}
\end{figure}

\begin{figure}[H]
  \centering
  \includegraphics[width=1\linewidth]{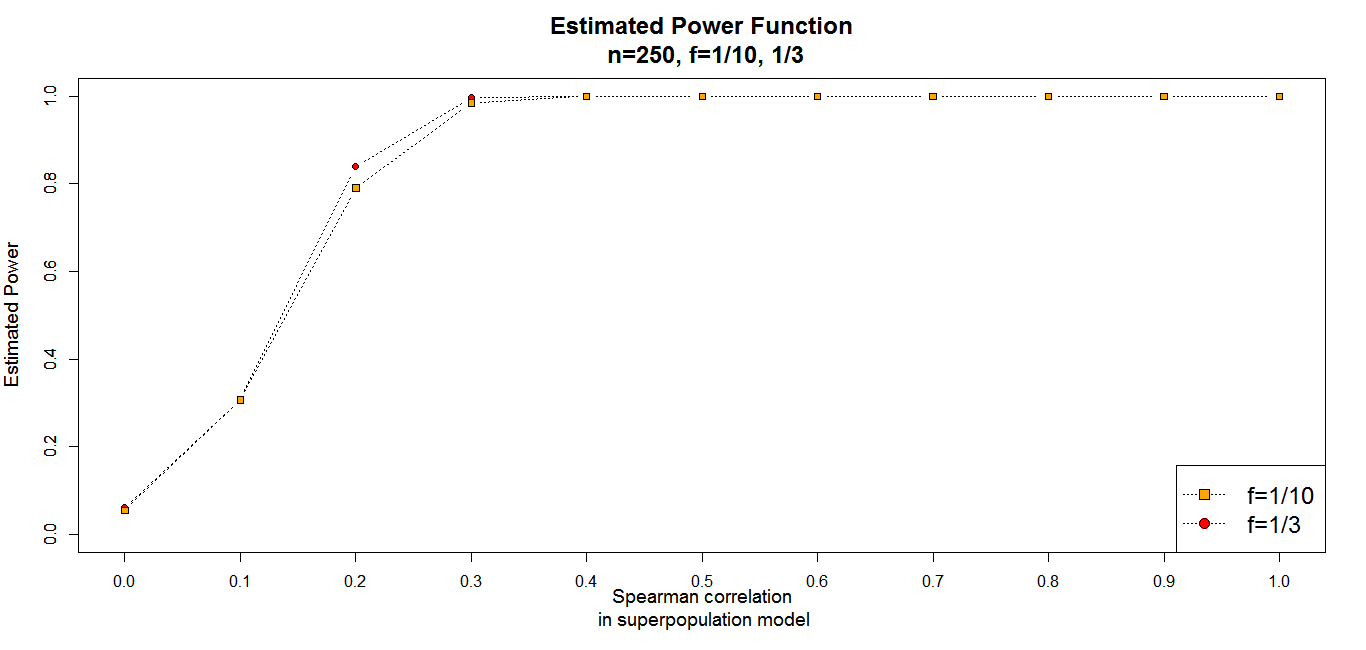}
  \caption{Estimated power function where $n=250, \ f=1/3,\ f=1/10.$}
  \label{fig:sub7}
\end{figure}
From table \ref{table:6}, it is seen that  estimated $\alpha$ is very close to the nominal $\alpha$. As expected, the largest error correspond to the smallest sample size ($n=50$) with a maximum absolute difference between $\alpha$ and $\hat{\alpha}$ of $2.4\%$. Of course, these errors decrease when the sample increases.
 As far as the sampling fractions are concerned, results in the cases $f=1/3$ and $f=1/10$ are similar; hence, the sampling fraction seems to play no special role.
The stimated power functions (figures $\ref{fig:sub5}$-$\ref{fig:sub7}$) exhibit a behavior similar to that of the estimated power functions studied in the previous section.
In fact, the estimated power function when $f=1/3$ dominates the estimated power function when $f=1/10$ for all sample sizes. Furthermore, differences between power functions decrease as the sample size increases. This suggests that the tests asymptotically have the same power,  whatever is the sampling fraction may be.

\newpage
\section{Appendix}

\begin{proof}[\textbf{Proof of Proposition \ref{main}}]
Claim $1$ is proved in \cite{contmarmec}. Claim $2$ is a consequence of Claim $1$ and Lemma 1.1 and Lemma 1.2 in \cite{r8}.
To prove Claim $3$, observe first that, from Donsker's Theorem (\cite{r19}, Th. 16.4, p. 141),
the process $W_N$ converges weakly to a Gaussian process $W_2$ with zero mean function and covariance kernel
\begin{equation}
C_2(y,t)=F(y\wedge t)-F(y)F(t)
\end{equation}
where the process $W_2 $ can be represented as
$(W_2(y)=B(F(y)),\ y\in \mathbb{R})$ where $B(t)$ is a Brownian bridge, i.e. a Browninan motion tied down to 0 when $t=1$.

To complete the proof, we only have to prove that the  asymptotic independence of the two sequences of processes $W_n^H$ and $W_N$.
To this purpose,
it is sufficient to show the asymptotic independence of their finite-dimensional distributions. Let $m, \, l$
be positive integers, and take points $m+l$ points $y_1^{(1)}, \ldots,y_m^{(1)},y_1^{(2)}, \dots,y_l^{(2)}$.
It is not difficult to see that
\begin{align*}
&\lim_{N\to\infty}Pr\left\lbrace W_n^H(y_1^{(1)})\leq z_1^{(1)},\ldots,W_n^H(y_m^{(1)})\leq z_m^{(1)},W_N(y_1^{(2)})\leq z_1^{(2)},\ldots,W_N(y_l^{(2)}) \leq z_l^{(2)}\right\rbrace= \\
&\lim_{N\to\infty}\E\left[I_{(W_n^H(y_1^{(1)})\leq z_1^{(1)},\ldots,W_n^H(y_m^{(1)})\leq z_m^{(1)},W_N(y_1^{(2)})\leq z_1^{(2)},\ldots,W_N(y_l^{(2)}) \leq z_l^{(2)})} \right]=\\
&\lim_{N\to\infty} \E_\MP\left[ \E_P[I_{(W_n^H(y_1^{(1)})\leq z_1^{(1)},\ldots,W_n^H(y_m^{(1)})\leq z_m^{(1)},W_N(y_1^{(2)})\leq z_1^{(2)},\ldots,W_N(y_l^{(2)}) \leq z_l^{(2)})}  \vert
\Y , \T] \right]=\\
&\lim_{N\to\infty}\E_\MP\left[ P\{W_n^H(y_1^{(1)})\leq z_1^{(1)},\ldots,W_n^H(y_m^{(1)})\leq z_m^{(1)}| \Y , \T
\} \cdot I_{(W_N(y_1^{(2)})\leq z_1^{(2)})}\cdots I_{(W_N(y_l^{(2)}) \leq z_l^{(2)})}  \right]=\\
&\E_\MP\left[\lim_{N\to\infty} P\{W_n^H(y_1^{(1)})\leq z_1^{(1)},\ldots,W_n^H(y_m^{(1)})\leq z_m^{(1)} | \Y , \T  \} \cdot\lim_{N\to\infty} I_{(W_N(y_1^{(2)})\leq z_1^{(2)})}\cdots I_{(W_N(y_l^{(2)}) \leq z_l^{(2)})}  \right]=\\
&Pr\{W_1(y_1^{(1)})\leq z_1^{(1)},\ldots,W_1(y_m^{(1)})\leq z_m^{(1)}\}\cdot \lim_{N\to\infty}\E_\MP\left[I_{(W_N(y_1^{(2)})\leq z_1^{(2)})}\cdots I_{(W_N(y_l^{(2)}) \leq z_l^{(2)})}  \right]=\\
&Pr\left \lbrace W_1(y_1^{(1)})\leq z_1^{(1)},\ldots,W_1(y_m^{(1)})\leq z_m^{(1)}\right \rbrace \cdot \lim_{N\to\infty}\MP\left \lbrace W_N(y_1^{(2)})\leq z_1^{(2)},\ldots,W_N(y_l^{(2)}) \leq z_l^{(2)}  \right \rbrace=\\
&Pr\left \lbrace W_1(y_1^{(1)})\leq z_1^{(1)},\ldots,W_1(y_m^{(1)})\leq z_m^{(1)}\right \rbrace Pr\left \lbrace W_2(y_1^{(2)})\leq z_1^{(2)},\ldots,W_2(y_l^{(2)}) \leq z_l^{(2)}  \right \rbrace
\end{align*}
which proves the asserted asymptotic independence.
\end{proof}

\begin{proof}[\textbf{Proof of Proposition \ref{glivcant}}]
First of all, from Proposition 1  and the Skorokhod representation theorem (cfr. \cite{r19}), it follows that
\begin{eqnarray}
\sup_y \left \vert \widehat{F}_H ( y) - F_N (y) \right \vert \rightarrow 0 \;\; {\mathrm{as}} \; N \rightarrow \infty
\label{gliv_1}
\end{eqnarray}
for a set of $\D$s with $P$-probability tending to 1, and for a set of (sequences of $Y_i$s, $T_{ij}$s having $\MP$-probability 1.
In the second place, from the ``classical'' Glivenko-Cantelli theorem, we have:
\begin{eqnarray}
\sup_y \left \vert F_H ( y) - F(y) \right \vert \rightarrow 0 \;\; {\mathrm{as}} \; N \rightarrow \infty
\label{gliv_2}
\end{eqnarray}
for a set of (sequences of) $Y_i$s, $T_{ij}$s having $\MP$-probability 1. Conclusion $( \ref{glivcantelli} )$ easily follows from
$( \ref{gliv_1} )$ and $( \ref{gliv_2} )$.
\end{proof}

\begin{Lemma} \label{lemma1}
Under the assumptions $H1$-$H6$, the quantity
\begin{eqnarray}
\frac{1}{N} \sum_{i=1}^{N} \frac{D_{i}}{\pi_{i}}  \label{ratio_1}
\end{eqnarray}
tends to $1$ as $N$ increases, for a set of (sequences of) $y_i$s, $t_{ij}$s having $\MP$-probability $1$, and for a set of $\D$s of
$P$-probability tending to $1$.
\end{Lemma}
\begin{proof}[\textbf{Proof}]
Conditionally on $\Y$, $\T$, the expectation of $( \ref{ratio_1} )$ w.r.t the sampling design $P$ is equal to 1.
The variance of $( \ref{ratio_1} )$ w.r.t. the sampling design $P$, conditionally on $\Y$, $\T$, is equal to
\begin{eqnarray}
\mathbb{V}_P \left (  \left . \frac{1}{N} \sum_{i=1}^{N} \frac{D_{i}}{\pi_{i}} \, \right \vert \Y , \, \T \right ) & = &
\frac{1}{N^{2}} \left \{ \sum_{i=1}^{N} \frac{1}{\pi_{i}^{2}} \mathbb{V}_P ( D_i \, \vert \Y , \, \T ) \right . \nonumber \\
\, & \, & + \left . \sum_{i=1}^{N} \sum_{j \neq i} \frac{1}{\pi_{i} \pi_{j}} \mathbb{C}_P ( D_i , \, D_J \, \vert  \Y , \, \T )
\right \}  \nonumber \\
\, & \leq & \frac{1}{N^{2}} \left \{ \sum_{i=1}^{N} \frac{1}{\pi_{i}} + \sum_{i=1}^{N} \sum_{j \neq i} \left \vert \frac{\pi_{ij} - \pi_{i} \pi_{j}}{\pi_{i} \pi_{j}} \right \vert \right \} . \nonumber
\end{eqnarray}
From $\pi_i = n x_i / \sum_{j=1}^{N} x_j$ (with $x_i = g( t_{i1}, \, \dots , \, t_{iL})$) and the strong law of large numbers, it is
 not difficult that the  $N^{-1} \sum_i \pi_i^{-1}$ converges for a set of (sequences of) $y_i$s,
$t_{ij}$s of $\MP$-probability $1$.
Furthermore, from the assumption of maximal asymptotic entropy of the sampling design implies (cfr. \cite{r16}, Th. 7.4) that
\begin{eqnarray}
\left \vert \frac{\pi_{ij} - \pi_{i} \pi_{j}}{\pi_{i} \pi_{j}} \right \vert \leq \frac{C}{N}
\nonumber
\end{eqnarray}
\noindent $C$ being an absolute constant. This shows that $( \ref{ratio_1} )$ tends to 1 as $N$ increases,
for a set of (sequences of)  $y_i$s, $t_{ij}$s of $\MP$-probability 1 and for a set of $\D$s of $P$-probability
tending to 1.
\end{proof}

\begin{proof}[\textbf{Proof of Proposition \ref{main*}}]
Tho prove Claim $1$, observe first that, as a consequence of Lemma $\ref{lemma1}$, we may write:
\begin{align}
&\E_{\MP^*}[N_i^*\vert \D ,\Y , \T]= \pi_i^{-1} D_i B_1 \\
&\mathbb{V}_{\MP^*}[N_i^*\vert \D ,\Y , \T] \leq \pi_i^{-1} D_i B_2 \\
&\vert \mathbb{C}_{\MP^*}[N_i^*,N_j^*\vert \D ,\Y , \T ] \vert \leq
c N^{-1} \pi_i^{-1} \pi_j^{-1}D_iD_j B3, \ j\neq i
\end{align}
where $B_1$ tends in $\MP$,$P$-probability to $1$ as $N$ goes to infinity, and $B_2$, $B_3$ are bounded in $\MP$,$P$-probability.
From Proposition 5 in \cite{contmarmec}, Claim $1$ follows. Claim $2$ is a consequence of Claim $1$ and Lemma 1.1 and Lemma 1.2 in \cite{r8}.

To prove Claim $3$, using the arguments in Th. 2.1 in \cite{r4} and Proposition $\ref{glivcant}$,
it follows that the process $W^{*}_N (y) = \sqrt{N} ( \widehat{F}^{*}_N (y) -
\widehat{F}_H (y))$ converges weakly to a Gaussian process of the form $B( F(y))$, $B$ being a Brownian bridge. Convergence
takes place for almost all $y_i$s, $t_{ij}$s, and for a set of $\D$s of $P$-probability tending to $1$. The asymptotic independence
between $W^{H*}_n (y)$ and $W^*_N (y)$ can be proved exactly as in Proposition  $\ref{main}$, from which $R1$ follows. $R2$ is a consequence
of the Hadamard differentiability of $\theta$.
\end{proof}

\begin{proof}[\textbf{Proof of Proposition \ref{Last-prop}}]
Let
\begin{equation*}
R_n^*(z)=P^*\{Z_{n,m}^*\leq z | s,\U^* \}
\end{equation*}
be the true (resampling) distribution function of $Z_{n,m}^*$ (defined in (\ref{Znm})). By the two sided Dvoretzky-Kiefer-Wolfowitz inequality (for more see \cite{r22} and \cite{r21}), we have
\begin{equation}
Pr\left\lbrace \sup_{z \in \mathbb{R}} \lvert \hat{R}^*_{n,M}(z)-R_n^*(z) \rvert>\epsilon \ \middle| \ s,\U \right\rbrace\leq 2e^{\{-2M\epsilon^2\}}.
\end{equation}
Taking into account that by Glivenko-Cantelli theorem (see Theorem 19.1 \cite{r9} p. 266) $R_n^*$ converges uniformly to $\Phi_{0,\sigma^2_{\theta}}$, and you have that (\ref{LP-1}) holds in probability. To obtain the almost sure convergence it is sufficient to use the Borel-Cantelli first lemma.
\end{proof}


\newpage
\nocite{*}

\begin{thebibliography}{1}





\bibitem{r10}
Antal, E. and Till{\'e}, Y. (2011). A direct bootstrap method for complex sampling designs from a finite population.
\emph{Journal of the American Statistical Association}, 106, 534--543.


\bibitem{r18}
Berger, Y. G. (1998). Rate of convergence to normal distribution for the Horvitz-Thompson estimator.
\emph{Journal of Statistical Planning and Inference}, 67, 209--226.



\bibitem{r4}
Bickel, P. J. and Freedman, D. A. (1981). Some asymptotic theory for the bootstrap.
\emph{The Annals of Statistics}, 9, 1196--1217.


\bibitem{r19}
Billingsley, P. (1968).
\emph{Convergence of probability measures}. Wiley, New York.



\bibitem{r24}
Boistard, H. and Lopuha{\"a}, H. P. and Ruiz-Gazen, A. (2015).
Functional central limit theorems in survey sampling.
\emph{ArXiv e-prints}, 1509.09273.

\bibitem{chatter11}
Chatterjee, A. (2011). Asymptotic properties of sample quantiles from a finite population.
\emph{Annals of the Institute of Statistical Mathematics},63, 157--159.


\bibitem{r12}
Cifarelli, D. M. and Conti, P. L. and Regazzini, E. (1996). On the asymptotic distribution of a general measure of monotone dependence. \emph{The Annals of Statistics}, 24, 1386--1399.


\bibitem{r1}
Cochran, W. G. (1939). The use of the analysis of variance in enumeration by sampling.
\emph{Journal of the American Statistical Association}, 34, 492--510.




\bibitem{r6}
Conti, P.L. (2014). On the estimation of the distribution function of a finite population under high entropy sampling designs, with applications. \emph{Sankhya B}, 76, 234--259.

\bibitem{contmarmec}
Conti, P.L,  Marella, D. and Mecatti, F. (2015). Recovering sampling distributions of statistics of finite populations via resampling: a predisctive approach. \emph{Submitted for publication}.


\bibitem{r5}
Conti, P. L. and Marella, D. (2015). Inference for Quantiles of a Finite Population: Asymptotic versus Resampling Results.
\emph{Scandinavian Journal of Statistics}, 42, 545--561.




\bibitem{r8}
Cs{\"o}rg{\H{o}}, S. and Rosalsky, A. (2003). A survey of limit laws for bootstrapped sums.
\emph{International Journal of Mathematics and Mathematical Sciences}, 45, 2835--2861.


\bibitem{r22}
Dvoretzky, Arye. and Kiefer, J.C. and Wolfowitz, J. (1956).
Asymptotic minimax character of the sample distribution function and of the classical multinomial estimator.
\emph{The Annals of Mathematical Statistics}, 27, 642--669.


\bibitem{r3}
Efron, B. (1979). Bootstrap methods: another look at the jackknife.
\emph{The Annals of Statistics}, 7, 1--26.



\bibitem{r17}
Grafstr{\"o}m, A. (2010). Entropy of unequal probability sampling designs.
\emph{Statistical Methodology}, 7, 84--97.

\bibitem{r23}
Gill, R. D. and Wellner, J. A. and Pr{\ae}stgaard, J. (1989). Non-and semi-parametric maximum likelihood estimators and the von {M}ises method ({P}art 1)(with discussion and reply),
\emph{Scandinavian Journal of Statistics}, 16, 97--128.



\bibitem{r15}
H{\'a}jek, J. (1964). Asymptotic theory of rejective sampling with varying probabilities from a finite population
\emph{The Annals of Mathematical Statistics}, 35, 1491--1523.


\bibitem{r16}
H{\'a}jek, J. and Dupac, V. (1981).
\emph{Sampling from a finite population}.Marcel Dekker, New York.

\bibitem{holmberg98}
Holmberg, A. (1998). A bootstrap approach to probability proportional-to-size sampling.
\emph{Proceedings of the ASA Section on Survey research Methods}, 378-383.


\bibitem{r25}
Kruskal, W.H. (1958). Ordinal measures of association. \emph{Journal of the American Statistical Association}, 53, 814--861.



\bibitem{r26}
Mai, J. F. and Scherer, M. (2012)
\emph{Simulating copulas: stochastic models, sampling algorithms and applications}. Imperial College Press, London.


\bibitem{r27}
Marshall, A.W. and Olkin, I. (1967). A multivariate exponential distribution.
\emph{Journal of the American Statistical Association}, 62, 30--44.



\bibitem{r28}
Marshall, A.W. and Olkin, I. (1967). A generalized bivariate exponential distribution.
\emph{Journal of Applied Probability}, 4, 291--302.

\bibitem{r21}
Massart, P. (1990). The tight constant in the {D}voretzky-{K}iefer-{W}olfowitz inequality.
\emph{The Annals of Probability}, 1269--1283.


\bibitem{r14}
Osier, G. (2009). Variance estimation for complex indicators of poverty and inequality using linearization techniques.
\emph{Survey Research Methods}, 3, 167--195.


\bibitem{pfeffer93}
Pfeffermann, D. (1993). The Role of Sampling Weights When Modeling Survey Data.
\emph{International Statistical Review}, 61, 317--337.




\bibitem{pfefsver04}
Pfeffermann, D. and Sverchkov, M. (2006). Prediction of finite population totals based on the sample distribution.
\emph{Survey Methodology}, 30, 79--92.



\bibitem{romano88}
Romano, J. P. (1988). A Bootstrap Revival of Some Nonparametric Distance Tests.
\emph{Journal of the American Statistical Association}, 83, 698--708.

\bibitem{romano89}
Romano, J. P. (1989). Bootstrap and Randomization Tests of some Nonparametric Hypotheses.\emph{The Annals of Statistics}, 17,
141--159.



\bibitem{r2}
S{\"a}rdnal, C. E. and Swensson, B. and Wretman, J. H. (1992).
\emph{Model Assisted Survey Sampling}, Springer-Verlag, New York.




\bibitem{r20}
Serfling, R. J. (1980). \emph{Approximation theorems of mathematical statistics}. Wiley, New York.


\bibitem{r9}
Van der Vaart, Aad W. (2000). \emph{Asymptotic statistics}. Cambridge University Press, Cambridge.

\bibitem{wang}
Wang, J. C. (2012). Sample distribution function based goodness-og-fit test for complex surveys.
\emph{Computational Statistics and Data Analysis}, 56, 664--679.


\end{thebibliography}
\end{document}